\documentclass[a4paper]{article}

\usepackage{palatino,subcaption,hyperref}
\usepackage{pgfarrows,tikz}
\usepackage[utf8]{inputenc}
\usepackage[T1]{fontenc}
\usepackage{amsthm,amsmath,amsfonts, amssymb,amscd,mathbbol,graphicx}
\newcommand{\Z}{\mathbb{Z}}
\newcommand{\N}{\mathbb{N}}

\usepackage{lmodern}
\newtheorem{lemma}{Lemma}
\newtheorem{claim}{Claim}
\newtheorem{corollary}{Corollary}
\newtheorem{definition}{Definition}
\newtheorem{theorem}{Theorem}

\newtheorem{example}{Example}

\usetikzlibrary{shapes,arrows}
\usetikzlibrary{patterns}
\usetikzlibrary{calc,backgrounds}
\usetikzlibrary{arrows.meta}
\usetikzlibrary{decorations.pathreplacing}
\usetikzlibrary{decorations.pathmorphing}
\definecolor{deepsaffron}{rgb}{1.0, 0.6, 0.2}
\definecolor{islamicgreen}{rgb}{0.0, 0.56, 0.0}
\definecolor{orange(webcolor)}{rgb}{1.0, 0.65, 0.0}
\definecolor{babyblue}{rgb}{0.54, 0.81, 0.94}
\definecolor{tomato}{rgb}{1.0, 0.39, 0.28}
\definecolor{green(ryb)}{rgb}{0.4, 0.69, 0.2}
\definecolor{azure(colorwheel)}{rgb}{0.0, 0.5, 1.0}

\emergencystretch=\maxdimen
\tikzset{
  double arrow/.style args={#1 colored by #2 and #3}{
    -stealth,line width=#1,#2, 
    postaction={draw,-stealth,#3,line width=(#1)-.5pt,
                shorten <=.5pt,shorten >=1pt}, 
  },
  double arrow debut/.style args={#1 colored by #2 and #3}{
    -stealth,line width=#1,#2, -, 
    postaction={draw,-stealth,#3,line width=(#1)-.5pt,
                shorten <=.3pt,shorten >=-.4pt, -}, 
  },
  double arrow milieu/.style args={#1 colored by #2 and #3}{
    -stealth,line width=#1,#2, -, 
    postaction={draw,-stealth,#3,line width=(#1)-.5pt,
                shorten <=-.1pt,shorten >=-.1pt, -}, 
  },
  double arrow fin/.style args={#1 colored by #2 and #3}{
    -stealth,line width=#1,#2, 
    postaction={draw,-stealth,#3,line width=(#1)-.5pt,
                shorten <=-.5pt,shorten >=.5pt}, 
  },
   paath/.pic = {
    \draw [ultra thick, draw = green, fill = green!50] (-2,0) -- (-2,1) -- (-1,1) -- (-1,2) 
    -- (0,2) -- (0,3) -- (1,3) -- (1,2) -- (2,2) -- (2,1) -- (3,1)-- (3,0) -- (2,0) -- (2,-1)
    -- (1,-1) -- (1,-2) -- (0,-2) -- (0,-1) -- (-1,-1) -- (-1,0)   -- cycle;     
  },
  paaath/.pic = {
    \draw [ultra thick, draw = green] (-2,0) -- (-2,1) -- (-1,1) -- (-1,2) 
    -- (0,2) -- (0,3) -- (1,3) -- (1,2) -- (2,2) -- (2,1) -- (3,1)-- (3,0) -- (2,0) -- (2,-1)
    -- (1,-1) -- (1,-2) -- (0,-2) -- (0,-1) -- (-1,-1) -- (-1,0)   -- cycle;     
  }
}

\begin{document}
\title{On Turedo Hierarchies and Intrinsic Universality}
\author{Samuel Nalin\thanks{Univ. Orléans, INSA Centre Val de Loire, LIFO EA 4022, Orléans, France} \and Guillaume Theyssier\thanks{I2M, Université Aix-Marseille, CNRS, Marseille, France}}

\maketitle
\begin{abstract}
  This paper is about turedos, which are Turing machine whose head can move in the plane (or in a higher-dimensional space) but only in a self-avoiding way, by putting marks (letters) on visited positions and moving only to unmarked, therefore unvisited, positions.
  The key parameter of turedos is their lookup radius: the distance up to which the head can look around in order to make its decision of where to move to and what mark to write.
  In this paper we study the hierarchy of turedos according to their lookup radius and the dimension of space using notions of simulation up to spatio-temporal rescaling (a standard approach in cellular automata or self-assembly systems).
  We establish that there is a rich interplay between the turedo parameters and the notion of simulation considered.
  We show in particular, for the most liberal simulations, the existence of 3D turedos of radius 1 that are intrinsically universal for all radii, but that this is impossible in dimension 2, where some radius 2 turedo are impossible to simulate at radius 1.
  Using stricter notions of simulation, intrinsic universality becomes impossible, even in dimension 3, and there is a strict radius hierarchy.
  Finally, when restricting to radius 1, universality is again possible in dimension 3, but not in dimension 2, where we show however that a radius 3 turedo can simulate all radius 1 turedos.
\end{abstract}

\section{Introduction}
\label{sec:intro}

The field of biomolecular computing has given rise to several theoretical models that describe growing process of (molecular) assemblies governed by local interaction or gluing rules.
One of the most studied one, the abstract Tile Assembly model (aTAM) \cite{Winfree_phd}, describes a process where the growth can happen anywhere asynchronously.
It was successfully implemented in vitro as DNA self-assembly \cite{RothemundPapadakisWinfree2004}.
In oritatami systems, introduced more recently \cite{GearyMeunierSchabanelSeki2019IJMS,GeMeScSe2016} and inspired from RNA origami \cite{GearyRothemundAndersen2014,GearyGMRA21}, the growth happens at a unique given point of the assembly and in a sequential manner.
Despite their obvious differences as models of biomolecular systems, they share common features as computational models.
Both are limited by the fact that they can only grow shapes (and not erase them), but both were shown capable of embedding universal computations in different ways \cite{GeMeScSe2018,GearyMSS18,CookFS11,LathropLPS11,Winfree_phd}.
Recently, it was also shown through new results on oritatami systems that the infinite limit shapes both models can generate from finite seeds have the same computational complexity and the same possible densities of occupied position in the plane \cite{STACS22}.
The key ingredients of these new results were, on one hand, the introduction of a new model, called \emph{turedos}, which abstracts away low level details of oritatami systems and is easier to program, and on the other hand, a proof that oritatami systems can simulate a large family of turedos.
This underlines the interest of turedos but also the surprising capabilities of models of sequential growth process in the plane.
The present work is entirely focused on the turedo model and aims at better understanding its limitations and its expressive power as a growth process and computational model.

\textbf{Turedos.} Intuitively, a turedo is a Turing machine whose head can move in the plane or in the 3-dimensional space but only in a self-avoiding way, \textit{i.e.} without going back to a previously visited position.
More precisely, the turedo's head can only move to empty or unmarked positions, and it must put a mark on each visited position when leaving it.
The key ingredient of a turedo, and what makes its main computational power, is its \emph{lookup radius}: when deciding a move and what mark to put on the visited position, the turedo has access to the local configuration of marks around its position, up to some finite distance.
Like any Turing machine, a turedo also has a set of internal head states.
Without entering into details, an oritatami system consists of a ``molecule'' made of ``beads'' that can attract each other. 
The molecule grows at each step following a periodic bead sequence and folds as follows: the $\delta$ most recently produced beads are free to move around to look for the position that maximizes the number of bonds they can make with each other ; then the first (oldest) bead among the $\delta$ most recent ones is fixed according to that position, a new free bead is added and the process iterates.
This behavior can be realized in a turedo of radius $\delta+1$.
The main result of \cite{STACS22} is that oritatami of delay $3$ can simulate turedos of radius $1$.
In fact, turedos also naturally capture variants of oritatami systems: for instance, one can imagine negative (repulsive) bonds in the process of maximizing gluing strength, or add local rules forbidding two bead type to be neighbor of each other.
Therefore negative results on turedos become negative results on oritatami systems, but also potential variants of them.

\textbf{Simulations and universality.} The main question we address here is how the capabilities of turedos change with their radius and what is the role of the dimension of space.
We are interested in qualitative differences and don't compare turedos move by move and cell by cell.
We rather use a notion of simulation allowing spatio-temporal rescaling, similar to the one used in cellular automata \cite{bulking2,BeckerMOT18}, in aTAM \cite{MeunierWoods2017,DLPSSW2012,PSTWW14} or in the simulation result of turedo by oritatamis \cite{STACS22}.
Intuitively a cell can be simulated by a block of cells, and a step by a finite number of steps.
We are naturally interested by hierarchy results (existence of turedos that can't be simulated by lower radius ones), or on the contrary by universality results (existence of turedos that can simulate all turedos or a large class of them).
For instance, the existence of intrinsically universal systems in the aTAM model was much studied and it was shown that it crucially depends on natural parameters of the model \cite{MeunierWoods2017,PSTWW14,DLPSSW2012,Hendricks_2016}.
However, as natural as it might seem, a formal notion of simulation is never a neutral choice when tackling these questions \cite{BeckerMOT18,bulking2,Boyer_2010}.
Thus, in addition to the parameters of the turedo model we also study the influence of the notion of simulation itself.
For this we identify various ingredients, in particular the possibility of fuzz, \textit{i.e.} the tolerance allowed for the simulating turedo to visit some regions of space close to the ongoing assembly that represent empty cells of the simulated turedo and are therefore not yet been visited by it.
Note that a similar notion of fuzz appeared in the intrinsic universality results on aTAM \cite{DLPSSW2012}.
We end up with three notions of simulation: the rigorous, the fuzzless and the liberal ones.

\textbf{Our contributions.} After formalizing all the concepts mentioned so far (Section~\ref{sec:def}), we establish a surprisingly diverse general picture of the capabilities of turedos of simulating each other. Separating the negative results (Section~\ref{sec:hier}) from the positive ones (Section~\ref{sec:univ}),
our results are the following:
\begin{enumerate}
\item under fuzzless simulation, intrinsic universality is impossible whatever the dimension, there is a radius hierarchy, and actually the impossibility strikes at radius $2$: no turedo can fuzzlessly simulate all radius $2$ turedos (Theorem~\ref{thm:nofuzznofun});
\item when restricting to radius $1$, rigorous intrinsic universality is possible in dimension $3$ (Theorem~\ref{thm:3Dunivr1rigor}), but not in dimension $2$ (Theorem~\ref{theo:noradius1univ});
\item however, we built a 2D turedo of radius 3 which is able to rigorously simulate all 2D turedo of radius 1 (Theorem~\ref{thm:heatsink});
\item for liberal simulations, we establish intrinsic universality in dimension 3 and a complete hierarchy collapse at radius 1 (Corollary~\ref{coro:fulluniversality3D});
\item finally we show that there is a 2D turedo of radius 2 which is impossible to simulate at radius 1, even under liberal simulations (Theorem~\ref{thm:radius2kolmojordan}).
\end{enumerate}

Besides the above results, our contribution also lies in the constructions and proof techniques.
For instance, the negative result 5 is based on a general lemma for 2D turedos of radius 1, which bounds the quantity of information (using Kolmogorov complexity) that can be carried from the seed to another connected component of the plane when the plane is divided by a 4-connected path.
As another example on the constructive side, result number 3 above uses a novel construction technique (called heat sink trick) that fully exploits the potential of radius 3 and could be used in any place where two unbounded streams of information have to be crossed.
Finally, in Section~\ref{sec:persp}, we discuss various questions left open and present what we believe are promising future research directions.

\section{Definitions}
\label{sec:def}

\newcommand\boule[2]{\ensuremath{B_{#1}(#2)}}
\newcommand\blanc{\ensuremath{\bot}}
\newcommand\unturedo{\ensuremath{T}}
\newcommand\globstat[1]{\mathcal{G}_{#1}}
\newcommand\globmap[1]{F_{#1}}

\newcommand\gst{\globstat{\unturedo}}
\newcommand\gma{\globmap{\unturedo}}
\newcommand\patt[3]{\ensuremath{#1[#2;#3]}}

We denote by $\N$ the set of natural numbers (including $0$), by $\N_+$ the positive ones, and by $\Z$ the set of integers.
We consider turedos on $\Z^d$ for ${d=2}$ or $3$ (and in particular, we don't use the hexagonal lattice on the plane, mostly to simplify notations and dimension change).
We fix a blank symbol \blanc used to represent empty positions and common to all turedos.
The ball of radius $r$ in dimension $d$, denoted \boule{d}{r}, is the set of positions reachable in $r$ elementary moves (moves along vector of the canonical based of $\Z^d$) from the origin.
\boule{d}{1} will always be the set of possible head moves in dimension $d$.
We denote by ${\patt{c}{z}{S}}$ the pattern of shape $S$ around position $z$ in configuration $c$, \textit{i.e.} the map ${z'\in S\mapsto c_{z+z'}}$.

\begin{definition}[Turedo]
  A \emph{turedo} of dimension $d$ and radius $r$ is a triple ${\unturedo=(A,Q,\delta)}$ where $A$ is its finite \emph{alphabet} with ${\blanc\in A}$, $Q$ is its finite set of \emph{head states} and \[\delta : Q\times A^{\boule{d}{r}}\to Q\times A\setminus\{\blanc\}\times \boule{d}{1}\] is its \emph{local transition map}. A \emph{global state} of $\unturedo$ is a triple ${(c,z,q)\in \gst = A^{\Z^d}\times\Z^d\times Q}$ specifying a configuration, a position and a head state. The \emph{global transition map} ${\gma:\gst\to\gst}$ associated to ${\unturedo}$ is defined by: 
  \[\gma(c,z,q) =
    \begin{cases}
      (c,z,q) &\text{ if $c(z)\neq\blanc$ or $c(z+\mu)\neq\blanc$}\\
      (c',z+\mu,q') &\text{ else,}
    \end{cases}
  \]
  where ${(q',a,\mu) = \delta(q,\patt{c}{z}{\boule{d}{r}})}$ and configuration ${c'}$ is defined by ${c'(z)=a}$ and ${c'(z')=c(z')}$ for all ${z'\neq z}$.
\end{definition}

\newcommand\dom[1]{\mathcal{D}(#1)}
The \emph{domain} of a global state ${(c,z,q)}$ is the set of non-blank positions of $c$ plus the head position, formally: ${\dom{c,z,q}=\{z\}\cup\bigcup\{z':c(z')\neq\blanc\}}$
A global state ${(c,z,q)}$ is finite if its domain is finite.
We are interested in orbits starting from finite initial global states, called finite seeds.

\begin{example}[The spiral-XOR turedo]\label{exa:spiralxor}
  Let ${A=\{\blanc,0,1\}}$ and ${Q = \{\leftarrow,\uparrow,\rightarrow,\downarrow\}}$.
  The \emph{spiral-XOR} turedo has the following local rule.
  The head holds a direction $d\in Q$ and tries to move in that direction and let behind, as letter of $A$, the sum modulo $2$ of the states of neighboring (non $\blanc$) positions.
  When it can do the $d$ move, it changes its internal state (counter-clockwise), when it can't it takes the first available move (clockwise) and doesn't change its state.
  Of course if there is no neighbor in state $\blanc$ then the turedo is blocked.
  The following table shows the local transition map up to rotation of $d$ (the red arrow can be rotated and all blue arrows are defined relatively to the red arrow):
  \begin{center}
    \begin{tabular}{c|c||c|c|c}
      state & 1st empty neighbor &new state &letter & move\\
      \hline
      \color{red}{$\uparrow$}&\color{blue}{$\uparrow$} (+0) &\color{blue}{$\leftarrow$} (-1) & $\sum\bmod 2$ & \color{blue}{$\uparrow$} (+0)\\
      \hline
      \color{red}{$\uparrow$}& \color{blue}{$\rightarrow$} (+1)&\color{blue}{$\uparrow$} (+0) & $\sum\bmod 2$ & \color{blue}{$\rightarrow$} (+1)\\
      \hline
      \color{red}{$\uparrow$}& \color{blue}{$\downarrow$} (+2)&\color{blue}{$\uparrow$} (+0) & $\sum\bmod 2$ & \color{blue}{$\downarrow$} (+2)\\
      \hline
      \color{red}{$\uparrow$}& \color{blue}{$\leftarrow$} (+3)&\color{blue}{$\uparrow$} (+0) & $\sum\bmod 2$ & \color{blue}{$\leftarrow$} (+3)\\
    \end{tabular}
  \end{center}
  See Figure~\ref{fig:spiralxor} for an example of orbit.
\end{example}

\newcommand\stateleft[2]{\fill[fill=black] (#1,#2)++(0,.5)--+(1,-.5)--+(1,.5)--cycle;} 
\newcommand\stateZZOO[2]{\stateZ{#1}{#2}} 
\newcommand\stateZ[2]{\fill[fill=black] (#1,#2)--+(1,.5)--+(0,1)--cycle;} 
\newcommand\stateZZOZO[2]{\fill[fill=green,fill opacity=.3] (#1,#2) rectangle +(1,1);}
\newcommand\stateOZOZO[2]{\fill[fill=yellow,fill opacity=.3] (#1,#2) rectangle +(1,1);}

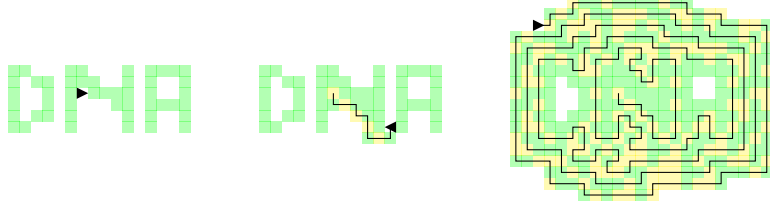
\begin{figure}
  \centering
  \begin{tikzpicture}[scale=.15]
    \begin{scope}[shift={(-47,0)}]
      \stateZZOZO{19}{22}\stateZZOZO{19}{23}\stateZZOZO{19}{24}\stateZZOZO{19}{25}\stateZZOZO{19}{26}\stateZZOZO{19}{27}\stateZZOZO{20}{22}\stateZZOZO{20}{27}\stateZZOZO{21}{23}\stateZZOZO{21}{26}\stateZZOZO{22}{23}\stateZZOZO{22}{24}\stateZZOZO{22}{25}\stateZZOZO{22}{26}\stateZZOZO{24}{22}\stateZZOZO{24}{23}\stateZZOZO{24}{24}\stateZZOZO{24}{25}\stateZZOZO{24}{26}\stateZZOZO{24}{27}\stateZ{25}{25}\stateZZOZO{25}{26}\stateZZOZO{25}{27}\stateZZOZO{26}{25}\stateZZOZO{26}{26}\stateZZOZO{27}{25}\stateZZOZO{28}{24}\stateZZOZO{28}{25}\stateZZOZO{29}{22}\stateZZOZO{29}{23}\stateZZOZO{29}{24}\stateZZOZO{29}{25}\stateZZOZO{29}{26}\stateZZOZO{29}{27}\stateZZOZO{31}{22}\stateZZOZO{31}{23}\stateZZOZO{31}{24}\stateZZOZO{31}{25}\stateZZOZO{31}{26}\stateZZOZO{31}{27}\stateZZOZO{32}{24}\stateZZOZO{32}{27}\stateZZOZO{33}{24}\stateZZOZO{33}{27}\stateZZOZO{34}{22}\stateZZOZO{34}{23}\stateZZOZO{34}{24}\stateZZOZO{34}{25}\stateZZOZO{34}{26}\stateZZOZO{34}{27}
      \draw (27,36) node {\it\small Initial global state};
    \end{scope}
    \begin{scope}[shift={(-25,0)}]
      \stateZZOZO{19}{22}\stateZZOZO{19}{23}\stateZZOZO{19}{24}\stateZZOZO{19}{25}\stateZZOZO{19}{26}\stateZZOZO{19}{27}\stateZZOZO{20}{22}\stateZZOZO{20}{27}\stateZZOZO{21}{23}\stateZZOZO{21}{26}\stateZZOZO{22}{23}\stateZZOZO{22}{24}\stateZZOZO{22}{25}\stateZZOZO{22}{26}\stateZZOZO{24}{22}\stateZZOZO{24}{23}\stateZZOZO{24}{24}\stateZZOZO{24}{25}\stateZZOZO{24}{26}\stateZZOZO{24}{27}\stateZZOZO{25}{26}\stateZZOZO{25}{27}\stateZZOZO{26}{25}\stateZZOZO{26}{26}\stateZZOZO{27}{25}\stateZZOZO{28}{24}\stateZZOZO{28}{25}\stateZZOZO{29}{22}\stateZZOZO{29}{23}\stateZZOZO{29}{24}\stateZZOZO{29}{25}\stateZZOZO{29}{26}\stateZZOZO{29}{27}\stateZZOZO{31}{22}\stateZZOZO{31}{23}\stateZZOZO{31}{24}\stateZZOZO{31}{25}\stateZZOZO{31}{26}\stateZZOZO{31}{27}\stateZZOZO{32}{24}\stateZZOZO{32}{27}\stateZZOZO{33}{24}\stateZZOZO{33}{27}\stateZZOZO{34}{22}\stateZZOZO{34}{23}\stateZZOZO{34}{24}\stateZZOZO{34}{25}\stateZZOZO{34}{26}\stateZZOZO{34}{27}
      \stateZZOZO{25}{24}\stateOZOZO{25}{25}\stateOZOZO{26}{24}\stateZZOZO{27}{24}\stateOZOZO{27}{23}\stateZZOZO{28}{23}\stateOZOZO{28}{22}\stateZZOZO{28}{21}\stateOZOZO{29}{21}\stateZZOZO{30}{21}
      \stateleft{30}{22}
      \begin{scope}[shift={(.5,.5)}]
        \draw[black] (25,25)--(25,24)--(25,24)--(26,24)--(26,24)--(27,24)--(27,24)--(27,23)--(27,23)--(28,23)--(28,23)--(28,22)--(28,22)--(28,21)--(28,21)--(29,21)--(29,21)--(30,21)--(30,21)--(30,22);
      \end{scope}
      \draw (27,36) node {\it\small Global state after 10 steps};
    \end{scope}
    \stateZZOZO{16}{19}\stateOZOZO{16}{20}\stateOZOZO{16}{21}\stateZZOZO{16}{22}\stateZZOZO{16}{23}\stateZZOZO{16}{24}\stateOZOZO{16}{25}\stateZZOZO{16}{26}\stateZZOZO{16}{27}\stateZZOZO{16}{28}\stateOZOZO{16}{29}\stateZZOZO{16}{30}\stateZZOZO{17}{19}\stateOZOZO{17}{20}\stateZZOZO{17}{21}\stateOZOZO{17}{22}\stateOZOZO{17}{23}\stateZZOZO{17}{24}\stateZZOZO{17}{25}\stateOZOZO{17}{26}\stateOZOZO{17}{27}\stateZZOZO{17}{28}\stateZZOZO{17}{29}\stateOZOZO{17}{30}\stateOZOZO{18}{19}\stateZZOZO{18}{20}\stateZZOZO{18}{21}\stateOZOZO{18}{22}\stateZZOZO{18}{23}\stateOZOZO{18}{24}\stateZZOZO{18}{25}\stateOZOZO{18}{26}\stateZZOZO{18}{27}\stateOZOZO{18}{28}\stateOZOZO{18}{29}\stateZZOZO{18}{30}\stateZZOO{18}{31}\stateOZOZO{19}{17}\stateZZOZO{19}{18}\stateOZOZO{19}{19}\stateZZOZO{19}{20}\stateOZOZO{19}{21}\stateZZOZO{19}{22}\stateZZOZO{19}{23}\stateZZOZO{19}{24}\stateZZOZO{19}{25}\stateZZOZO{19}{26}\stateZZOZO{19}{27}\stateZZOZO{19}{28}\stateOZOZO{19}{29}\stateZZOZO{19}{30}\stateOZOZO{19}{31}\stateZZOZO{19}{32}\stateZZOZO{20}{17}\stateZZOZO{20}{18}\stateOZOZO{20}{19}\stateZZOZO{20}{20}\stateZZOZO{20}{21}\stateZZOZO{20}{22}\stateZZOZO{20}{27}\stateOZOZO{20}{28}\stateZZOZO{20}{29}\stateZZOZO{20}{30}\stateZZOZO{20}{31}\stateOZOZO{20}{32}\stateZZOZO{21}{17}\stateOZOZO{21}{18}\stateZZOZO{21}{19}\stateOZOZO{21}{20}\stateOZOZO{21}{21}\stateZZOZO{21}{22}\stateZZOZO{21}{23}\stateZZOZO{21}{26}\stateZZOZO{21}{27}\stateZZOZO{21}{28}\stateZZOZO{21}{29}\stateOZOZO{21}{30}\stateOZOZO{21}{31}\stateZZOZO{21}{32}\stateOZOZO{21}{33}\stateZZOZO{22}{16}\stateOZOZO{22}{17}\stateZZOZO{22}{18}\stateZZOZO{22}{19}\stateOZOZO{22}{20}\stateZZOZO{22}{21}\stateOZOZO{22}{22}\stateZZOZO{22}{23}\stateZZOZO{22}{24}\stateZZOZO{22}{25}\stateZZOZO{22}{26}\stateOZOZO{22}{27}\stateOZOZO{22}{28}\stateZZOZO{22}{29}\stateZZOZO{22}{30}\stateOZOZO{22}{31}\stateOZOZO{22}{32}\stateZZOZO{22}{33}\stateOZOZO{23}{16}\stateZZOZO{23}{17}\stateOZOZO{23}{18}\stateOZOZO{23}{19}\stateZZOZO{23}{20}\stateOZOZO{23}{21}\stateZZOZO{23}{22}\stateOZOZO{23}{23}\stateZZOZO{23}{24}\stateOZOZO{23}{25}\stateZZOZO{23}{26}\stateOZOZO{23}{27}\stateZZOZO{23}{28}\stateOZOZO{23}{29}\stateOZOZO{23}{30}\stateOZOZO{23}{31}\stateZZOZO{23}{32}\stateZZOZO{23}{33}\stateZZOZO{24}{16}\stateZZOZO{24}{17}\stateOZOZO{24}{18}\stateZZOZO{24}{19}\stateOZOZO{24}{20}\stateZZOZO{24}{21}\stateZZOZO{24}{22}\stateZZOZO{24}{23}\stateZZOZO{24}{24}\stateZZOZO{24}{25}\stateZZOZO{24}{26}\stateZZOZO{24}{27}\stateOZOZO{24}{28}\stateZZOZO{24}{29}\stateZZOZO{24}{30}\stateZZOZO{24}{31}\stateZZOZO{24}{32}\stateOZOZO{24}{33}\stateOZOZO{25}{16}\stateZZOZO{25}{17}\stateOZOZO{25}{18}\stateOZOZO{25}{19}\stateZZOZO{25}{20}\stateOZOZO{25}{21}\stateZZOZO{25}{22}\stateOZOZO{25}{23}\stateZZOZO{25}{24}\stateOZOZO{25}{25}\stateZZOZO{25}{26}\stateZZOZO{25}{27}\stateZZOZO{25}{28}\stateZZOZO{25}{29}\stateOZOZO{25}{30}\stateOZOZO{25}{31}\stateOZOZO{25}{32}\stateZZOZO{25}{33}\stateOZOZO{26}{16}\stateOZOZO{26}{17}\stateZZOZO{26}{18}\stateOZOZO{26}{19}\stateOZOZO{26}{20}\stateZZOZO{26}{21}\stateOZOZO{26}{22}\stateZZOZO{26}{23}\stateOZOZO{26}{24}\stateZZOZO{26}{25}\stateZZOZO{26}{26}\stateZZOZO{26}{27}\stateOZOZO{26}{28}\stateOZOZO{26}{29}\stateZZOZO{26}{30}\stateOZOZO{26}{31}\stateOZOZO{26}{32}\stateZZOZO{26}{33}\stateZZOZO{27}{16}\stateZZOZO{27}{17}\stateZZOZO{27}{18}\stateOZOZO{27}{19}\stateOZOZO{27}{20}\stateOZOZO{27}{21}\stateZZOZO{27}{22}\stateOZOZO{27}{23}\stateZZOZO{27}{24}\stateZZOZO{27}{25}\stateZZOZO{27}{26}\stateOZOZO{27}{27}\stateZZOZO{27}{28}\stateOZOZO{27}{29}\stateZZOZO{27}{30}\stateZZOZO{27}{31}\stateOZOZO{27}{32}\stateZZOZO{27}{33}\stateOZOZO{28}{16}\stateOZOZO{28}{17}\stateOZOZO{28}{18}\stateZZOZO{28}{19}\stateZZOZO{28}{20}\stateZZOZO{28}{21}\stateOZOZO{28}{22}\stateZZOZO{28}{23}\stateZZOZO{28}{24}\stateZZOZO{28}{25}\stateOZOZO{28}{26}\stateZZOZO{28}{27}\stateOZOZO{28}{28}\stateOZOZO{28}{29}\stateZZOZO{28}{30}\stateOZOZO{28}{31}\stateZZOZO{28}{32}\stateZZOZO{28}{33}\stateOZOZO{29}{17}\stateOZOZO{29}{18}\stateOZOZO{29}{19}\stateZZOZO{29}{20}\stateOZOZO{29}{21}\stateZZOZO{29}{22}\stateZZOZO{29}{23}\stateZZOZO{29}{24}\stateZZOZO{29}{25}\stateZZOZO{29}{26}\stateZZOZO{29}{27}\stateZZOZO{29}{28}\stateZZOZO{29}{29}\stateOZOZO{29}{30}\stateZZOZO{29}{31}\stateZZOZO{29}{32}\stateOZOZO{29}{33}\stateOZOZO{30}{17}\stateOZOZO{30}{18}\stateOZOZO{30}{19}\stateOZOZO{30}{20}\stateZZOZO{30}{21}\stateOZOZO{30}{22}\stateZZOZO{30}{23}\stateOZOZO{30}{24}\stateZZOZO{30}{25}\stateOZOZO{30}{26}\stateZZOZO{30}{27}\stateOZOZO{30}{28}\stateOZOZO{30}{29}\stateZZOZO{30}{30}\stateZZOZO{30}{31}\stateOZOZO{30}{32}\stateZZOZO{30}{33}\stateZZOZO{31}{17}\stateZZOZO{31}{18}\stateZZOZO{31}{19}\stateZZOZO{31}{20}\stateOZOZO{31}{21}\stateZZOZO{31}{22}\stateZZOZO{31}{23}\stateZZOZO{31}{24}\stateZZOZO{31}{25}\stateZZOZO{31}{26}\stateZZOZO{31}{27}\stateZZOZO{31}{28}\stateOZOZO{31}{29}\stateZZOZO{31}{30}\stateOZOZO{31}{31}\stateOZOZO{31}{32}\stateZZOZO{31}{33}\stateOZOZO{32}{17}\stateZZOZO{32}{18}\stateOZOZO{32}{19}\stateZZOZO{32}{20}\stateZZOZO{32}{21}\stateOZOZO{32}{22}\stateZZOZO{32}{23}\stateZZOZO{32}{24}\stateZZOZO{32}{27}\stateZZOZO{32}{28}\stateOZOZO{32}{29}\stateZZOZO{32}{30}\stateZZOZO{32}{31}\stateZZOZO{32}{32}\stateOZOZO{33}{17}\stateOZOZO{33}{18}\stateZZOZO{33}{19}\stateOZOZO{33}{20}\stateZZOZO{33}{21}\stateOZOZO{33}{22}\stateOZOZO{33}{23}\stateZZOZO{33}{24}\stateZZOZO{33}{27}\stateOZOZO{33}{28}\stateOZOZO{33}{29}\stateZZOZO{33}{30}\stateOZOZO{33}{31}\stateOZOZO{33}{32}\stateZZOZO{34}{17}\stateZZOZO{34}{18}\stateZZOZO{34}{19}\stateOZOZO{34}{20}\stateOZOZO{34}{21}\stateZZOZO{34}{22}\stateZZOZO{34}{23}\stateZZOZO{34}{24}\stateZZOZO{34}{25}\stateZZOZO{34}{26}\stateZZOZO{34}{27}\stateZZOZO{34}{28}\stateOZOZO{34}{29}\stateZZOZO{34}{30}\stateZZOZO{34}{31}\stateOZOZO{34}{32}\stateZZOZO{35}{18}\stateOZOZO{35}{19}\stateZZOZO{35}{20}\stateZZOZO{35}{21}\stateOZOZO{35}{22}\stateZZOZO{35}{23}\stateOZOZO{35}{24}\stateZZOZO{35}{25}\stateOZOZO{35}{26}\stateZZOZO{35}{27}\stateOZOZO{35}{28}\stateZZOZO{35}{29}\stateZZOZO{35}{30}\stateZZOZO{35}{31}\stateZZOZO{36}{18}\stateOZOZO{36}{19}\stateOZOZO{36}{20}\stateZZOZO{36}{21}\stateZZOZO{36}{22}\stateOZOZO{36}{23}\stateOZOZO{36}{24}\stateZZOZO{36}{25}\stateZZOZO{36}{26}\stateOZOZO{36}{27}\stateOZOZO{36}{28}\stateZZOZO{36}{29}\stateOZOZO{36}{30}\stateOZOZO{36}{31}\stateOZOZO{37}{18}\stateZZOZO{37}{19}\stateZZOZO{37}{20}\stateOZOZO{37}{21}\stateZZOZO{37}{22}\stateZZOZO{37}{23}\stateZZOZO{37}{24}\stateOZOZO{37}{25}\stateZZOZO{37}{26}\stateZZOZO{37}{27}\stateZZOZO{37}{28}\stateOZOZO{37}{29}\stateZZOZO{37}{30}\stateOZOZO{37}{31}\stateZZOZO{38}{18}\stateOZOZO{38}{19}\stateZZOZO{38}{20}\stateZZOZO{38}{21}\stateOZOZO{38}{22}\stateZZOZO{38}{23}\stateOZOZO{38}{24}\stateOZOZO{38}{25}\stateZZOZO{38}{26}\stateOZOZO{38}{27}\stateZZOZO{38}{28}\stateZZOZO{38}{29}\stateOZOZO{38}{30}\stateZZOZO{38}{31}
    \begin{scope}[shift={(.5,.5)}]
      \draw[black] (25,25)--(25,24)--(25,24)--(26,24)--(26,24)--(27,24)--(27,24)--(27,23)--(27,23)--(28,23)--(28,23)--(28,22)--(28,22)--(28,21)--(28,21)--(29,21)--(29,21)--(30,21)--(30,21)--(30,22)--(30,22)--(30,23)--(30,23)--(30,24)--(30,24)--(30,25)--(30,25)--(30,26)--(30,26)--(30,27)--(30,27)--(30,28)--(30,28)--(29,28)--(29,28)--(28,28)--(28,28)--(28,27)--(28,27)--(28,26)--(28,26)--(27,26)--(27,26)--(27,27)--(27,27)--(26,27)--(26,27)--(26,28)--(26,28)--(27,28)--(27,28)--(27,29)--(27,29)--(26,29)--(26,29)--(25,29)--(25,29)--(25,28)--(25,28)--(24,28)--(24,28)--(23,28)--(23,28)--(23,27)--(23,27)--(23,26)--(23,26)--(23,25)--(23,25)--(23,24)--(23,24)--(23,23)--(23,23)--(23,22)--(23,22)--(23,21)--(23,21)--(24,21)--(24,21)--(25,21)--(25,21)--(25,22)--(25,22)--(25,23)--(25,23)--(26,23)--(26,23)--(26,22)--(26,22)--(27,22)--(27,22)--(27,21)--(27,21)--(26,21)--(26,21)--(26,20)--(26,20)--(27,20)--(27,20)--(28,20)--(28,20)--(29,20)--(29,20)--(30,20)--(30,20)--(31,20)--(31,20)--(31,21)--(31,21)--(32,21)--(32,21)--(32,22)--(32,22)--(32,23)--(32,23)--(33,23)--(33,23)--(33,22)--(33,22)--(33,21)--(33,21)--(34,21)--(34,21)--(35,21)--(35,21)--(35,22)--(35,22)--(35,23)--(35,23)--(35,24)--(35,24)--(35,25)--(35,25)--(35,26)--(35,26)--(35,27)--(35,27)--(35,28)--(35,28)--(34,28)--(34,28)--(33,28)--(33,28)--(32,28)--(32,28)--(31,28)--(31,28)--(31,29)--(31,29)--(30,29)--(30,29)--(29,29)--(29,29)--(28,29)--(28,29)--(28,30)--(28,30)--(27,30)--(27,30)--(26,30)--(26,30)--(25,30)--(25,30)--(24,30)--(24,30)--(24,29)--(24,29)--(23,29)--(23,29)--(22,29)--(22,29)--(22,28)--(22,28)--(22,27)--(22,27)--(21,27)--(21,27)--(21,28)--(21,28)--(20,28)--(20,28)--(19,28)--(19,28)--(18,28)--(18,28)--(18,27)--(18,27)--(18,26)--(18,26)--(18,25)--(18,25)--(18,24)--(18,24)--(18,23)--(18,23)--(18,22)--(18,22)--(18,21)--(18,21)--(19,21)--(19,21)--(20,21)--(20,21)--(21,21)--(21,21)--(21,22)--(21,22)--(22,22)--(22,22)--(22,21)--(22,21)--(22,20)--(22,20)--(21,20)--(21,20)--(21,19)--(21,19)--(22,19)--(22,19)--(23,19)--(23,19)--(23,20)--(23,20)--(24,20)--(24,20)--(25,20)--(25,20)--(25,19)--(25,19)--(24,19)--(24,19)--(24,18)--(24,18)--(25,18)--(25,18)--(26,18)--(26,18)--(26,19)--(26,19)--(27,19)--(27,19)--(28,19)--(28,19)--(29,19)--(29,19)--(30,19)--(30,19)--(31,19)--(31,19)--(32,19)--(32,19)--(32,20)--(32,20)--(33,20)--(33,20)--(34,20)--(34,20)--(35,20)--(35,20)--(36,20)--(36,20)--(36,21)--(36,21)--(36,22)--(36,22)--(36,23)--(36,23)--(36,24)--(36,24)--(36,25)--(36,25)--(36,26)--(36,26)--(36,27)--(36,27)--(36,28)--(36,28)--(36,29)--(36,29)--(35,29)--(35,29)--(34,29)--(34,29)--(33,29)--(33,29)--(32,29)--(32,29)--(32,30)--(32,30)--(31,30)--(31,30)--(30,30)--(30,30)--(29,30)--(29,30)--(29,31)--(29,31)--(28,31)--(28,31)--(27,31)--(27,31)--(26,31)--(26,31)--(25,31)--(25,31)--(24,31)--(24,31)--(23,31)--(23,31)--(23,30)--(23,30)--(22,30)--(22,30)--(21,30)--(21,30)--(21,29)--(21,29)--(20,29)--(20,29)--(19,29)--(19,29)--(18,29)--(18,29)--(17,29)--(17,29)--(17,28)--(17,28)--(17,27)--(17,27)--(17,26)--(17,26)--(17,25)--(17,25)--(17,24)--(17,24)--(17,23)--(17,23)--(17,22)--(17,22)--(17,21)--(17,21)--(17,20)--(17,20)--(18,20)--(18,20)--(19,20)--(19,20)--(20,20)--(20,20)--(20,19)--(20,19)--(20,18)--(20,18)--(21,18)--(21,18)--(22,18)--(22,18)--(23,18)--(23,18)--(23,17)--(23,17)--(24,17)--(24,17)--(25,17)--(25,17)--(26,17)--(26,17)--(27,17)--(27,17)--(27,18)--(27,18)--(28,18)--(28,18)--(29,18)--(29,18)--(30,18)--(30,18)--(31,18)--(31,18)--(32,18)--(32,18)--(33,18)--(33,18)--(33,19)--(33,19)--(34,19)--(34,19)--(35,19)--(35,19)--(36,19)--(36,19)--(37,19)--(37,19)--(37,20)--(37,20)--(37,21)--(37,21)--(37,22)--(37,22)--(37,23)--(37,23)--(37,24)--(37,24)--(37,25)--(37,25)--(37,26)--(37,26)--(37,27)--(37,27)--(37,28)--(37,28)--(37,29)--(37,29)--(37,30)--(37,30)--(36,30)--(36,30)--(35,30)--(35,30)--(34,30)--(34,30)--(33,30)--(33,30)--(33,31)--(33,31)--(32,31)--(32,31)--(31,31)--(31,31)--(30,31)--(30,31)--(30,32)--(30,32)--(29,32)--(29,32)--(28,32)--(28,32)--(27,32)--(27,32)--(26,32)--(26,32)--(25,32)--(25,32)--(24,32)--(24,32)--(23,32)--(23,32)--(22,32)--(22,32)--(22,31)--(22,31)--(21,31)--(21,31)--(20,31)--(20,31)--(20,30)--(20,30)--(19,30)--(19,30)--(18,30)--(18,30)--(17,30)--(17,30)--(16,30)--(16,30)--(16,29)--(16,29)--(16,28)--(16,28)--(16,27)--(16,27)--(16,26)--(16,26)--(16,25)--(16,25)--(16,24)--(16,24)--(16,23)--(16,23)--(16,22)--(16,22)--(16,21)--(16,21)--(16,20)--(16,20)--(16,19)--(16,19)--(17,19)--(17,19)--(18,19)--(18,19)--(19,19)--(19,19)--(19,18)--(19,18)--(19,17)--(19,17)--(20,17)--(20,17)--(21,17)--(21,17)--(22,17)--(22,17)--(22,16)--(22,16)--(23,16)--(23,16)--(24,16)--(24,16)--(25,16)--(25,16)--(26,16)--(26,16)--(27,16)--(27,16)--(28,16)--(28,16)--(28,17)--(28,17)--(29,17)--(29,17)--(30,17)--(30,17)--(31,17)--(31,17)--(32,17)--(32,17)--(33,17)--(33,17)--(34,17)--(34,17)--(34,18)--(34,18)--(35,18)--(35,18)--(36,18)--(36,18)--(37,18)--(37,18)--(38,18)--(38,18)--(38,19)--(38,19)--(38,20)--(38,20)--(38,21)--(38,21)--(38,22)--(38,22)--(38,23)--(38,23)--(38,24)--(38,24)--(38,25)--(38,25)--(38,26)--(38,26)--(38,27)--(38,27)--(38,28)--(38,28)--(38,29)--(38,29)--(38,30)--(38,30)--(38,31)--(38,31)--(37,31)--(37,31)--(36,31)--(36,31)--(35,31)--(35,31)--(34,31)--(34,31)--(34,32)--(34,32)--(33,32)--(33,32)--(32,32)--(32,32)--(31,32)--(31,32)--(31,33)--(31,33)--(30,33)--(30,33)--(29,33)--(29,33)--(28,33)--(28,33)--(27,33)--(27,33)--(26,33)--(26,33)--(25,33)--(25,33)--(24,33)--(24,33)--(23,33)--(23,33)--(22,33)--(22,33)--(21,33)--(21,33)--(21,32)--(21,32)--(20,32)--(20,32)--(19,32)--(19,32)--(19,31)--(19,31)--(18,31)--(18,31);
    \end{scope}
    \draw (27,36) node {\it\small Global state after 307 steps};
  \end{tikzpicture}
  \caption{\label{fig:spiralxor} Example of orbit of the spiral-XOR turedo (example~\ref{exa:spiralxor}) starting from a finite seed. Green represents $0$, yellow represents $1$ and white represents $\blanc$. The head holds a direction (where the black triangle is pointing to) and its self-avoiding trajectory since the beginning is drawn as a black path.}
\end{figure}

The main focus of this paper is to understand the role of radius and dimension in the computational complexity of turedos.
\newcommand\tur[2]{\ensuremath{\textsc{TUR}_{#1}(#2)}}
\newcommand\turdim[1]{\ensuremath{\textsc{TUR}_{#1}}}
We will denote by $\tur{d}{r}$ the set of turedos of dimension $d$ and radius $r$, and ${\turdim{d}=\bigcup_{r\geq 1}\tur{d}{r}}$.

Before formalizing simulation, we first need to define \emph{block encodings} which are ways to represent global states of a simulated turedo by blocks in the simulator.
\newcommand\bofs[2]{\ensuremath{\mu_{#1}(#2)}}
\newcommand\brnd[2]{\ensuremath{\rho_{#1}(#2)}}
\newcommand\brect[1]{\ensuremath{R_{#1}}}
Any given ${b\in\N_+^d}$ defines a rectangular block ${\brect{b} = \{z\in\N^d : 0\leq z_i<b_i\text{ for $1\leq i\leq d$}\}}$ and ${Z^d}$ can be tiled by translated copies of $\brect{b}$ in a regular way by placing them on the sublattice ${b\otimes\Z^d}$ where $\otimes$ denotes the component-wise product.
Each position ${z\in\Z^d}$ can be uniquely decomposed into ${z = \brnd{b}{z} + \bofs{b}{z}}$ where ${\brnd{b}{z}\in b\otimes\Z^d}$ is the reference point of a block and ${\bofs{b}{z}\in \brect{b}}$ is an offset inside it.

\begin{definition}[Block encoding]
  Let us fix a dimension $d$.
  Given two pairs of alphabets and state sets $(A_1,Q_1)$ and  $(A_2,Q_2)$ with ${\blanc\in A_1\cap A_2}$, a \emph{block encoding} of global states ${\mathcal{G}_1=A_1^{\Z^d}\times\Z^d\times Q_1}$ into ${\mathcal{G}_2=A_2^{\Z^d}\times\Z^d\times Q_2}$ is given by a \emph{block size} ${b\in\N_+^d}$ and two \emph{partial onto} maps:
  \begin{itemize}
  \item the \emph{headless block decoding map} ${\alpha : D_\alpha\subseteq A_2^{\brect{b}}\to A_1}$ verifying ${\blanc^{\brect{b}}\in D_\alpha}$ and ${\alpha(\blanc^{\brect{b}})=\blanc}$,
  \item the \emph{head block decoding map} ${\beta : D_\beta\subseteq A_2^{\brect{b}}\times \brect{b}\times Q_2\to Q_1\times A_1}$.
  \end{itemize}
  A global state ${(c,z,q)\in \mathcal{G}_2}$ is \emph{valid} for the encoding if it is made only of patterns from $D_\alpha$ far from the head and $D_\beta$ around the head, precisely if: ${(\patt{c}{\rho_b(z)}{\brect{b}},\bofs{b}{z},q)\in D_\beta}$ and ${\patt{c}{\brnd{b}{z'}}{\brect{b}}\in D_\alpha}$ for all ${z'\in\Z^d}$ such that ${\brnd{b}{z'}\neq\brnd{b}{z}}$.

  Finally, the \emph{global decoding map} $\Gamma$ associates to any valid global state ${(c_2,z_2,q_2)\in \mathcal{G}_2}$ a global state ${(c_1,z_1,q_1)\in \mathcal{G}_1}$ defined by application of decoding maps $\alpha$ or $\beta$ on each block according to the presence of the head in the block, \textit{i.e.} :
  \begin{itemize}
  \item ${b\otimes z_1=\brnd{b}{z_2}}$,
  \item ${(q_1,c_1(z_1)) = \beta(\patt{c_2}{\brnd{b}{z_1}}{\brect{b}},\bofs{b}{z_2},q_2)}$,
  \item ${c_1(z)=\alpha(\patt{c_2}{\brnd{b}{z_2}}{\brect{b}})}$ for all ${z\neq z_1}$.
  \end{itemize}
\end{definition}

The map $\alpha$ and $\beta$ being partial and onto intuitively means that not all global states are valid, and that any global state of ${A_1^{\Z^d}\times\Z^d\times Q_1}$ can be encoded.
\newcommand\bdom[2]{\mathcal{D}_{#1}(#2)}
Note that the headless block decoding map $\alpha$ always decodes blank blocks ${\blanc^{\brect{b}}}$ as blank state $\blanc$.
Denote by ${\bdom{b}{c_2,z_2,q_2}}$ the \emph{block domain} of global state ${(c_2,z_2,q_2)}$ which is the set of blocks that are not entirely blank, \textit{i.e.} 
${\bdom{b}{c_2,z_2,q_2} = \{z: b\otimes z=\brnd{b}{z_2}\text{ or }\patt{c_2}{b\otimes z}{\brect{b}}\neq\blanc^{\brect{b}}\}}$.

\newcommand\turun{\unturedo_1}
\newcommand\gstaun{\globstat{\turun}}
\newcommand\gmaun{\globmap{\turun}}
\newcommand\turdeux{\unturedo_2}
\newcommand\gstadeux{\globstat{\turdeux}}
\newcommand\gmadeux{\globmap{\turdeux}}

We can now define simulations precisely using block encodings.
Intuitively, we ask for the simulator to be able to reproduce any orbit of the simulated turedo starting from a finite seed, and using a (fixed) finite number of steps to simulate one step.
Since we want the initial seed of the simulator to be neutral and without any pre-computed information about the future of the simulated orbit, we ask that its block domain correspond to the domain of the simulated seed.

\begin{definition}[Simulation]
  Let $d$ be a fixed dimension. We say that a $d$-dimensional turedo $\turdeux$ \emph{simulates} a $d$-dimensional turedo $\turun$ if there is:
  \begin{itemize}
  \item a block encoding of $\gstaun$ into $\gstadeux$ of bock size $b$ and global decoding map $\Gamma$,
  \item a time scaling factor $k\in\N_+$,
  \end{itemize}
  such that for each finite global state ${(c_1,z_1,q_1)\in\gstaun}$ and each global state ${(c_2,z_2,q_2)\in\gstadeux}$ verifying:
  \begin{itemize}
  \item corresponding block domain: ${\dom{c_1,z_1,q_1}=\bdom{b}{c_2,z_2,q_2}}$,
  \item correct encoding: ${(c_1,z_1,q_1)=\Gamma(c_2,z_2,q_2)}$,
  \end{itemize}
  then it holds ${\forall t\in\N, \gmaun^t(c_1,z_1,q_1) = \Gamma\bigl(\gmadeux^{kt}(c_2,z_2,q_2))}$.
  Such simulations are the base upon which we define three variants (two restrictions and one generalization).
  
  We say that a simulation is \emph{fuzzless} if the block domain in the simulator orbit remains identical to the domain of the simulated orbit, precisely: 
  ${\forall t\in\N : \bdom{b}{\gmadeux^{kt}(c_2,z_2,q_2)}=\dom{\gmaun^t(c_1,z_1,q_1)}.}$

  We say that a simulation is \emph{rigorous} if the movements of the head of $\turdeux$ in simulating orbits strictly remains inside blocks corresponding to the simulated head position of $\turun$, even at intermediate steps, precisely: if $z_1^t$ denotes the head position of $\gmaun^t(c_1,z_1,q_1)$ and $z_2^t$ that of $\gmadeux^t(c_2,z_2,q_2)$, it holds for all $t'$ with
  ${kt\leq t'\leq k(t+1): z_2^{t'}\in (b\otimes z_1^t+\brect{b})\cup (b\otimes z_1^{t+1}+\brect{b})}$.

  Finally, a \emph{liberal simulation} is a generalized simulation where we only ask that for each finite global state ${(c_1,z_1,q_1)\in\gstaun}$ there exists a global state ${(c_2,z_2,q_2)\in\gstadeux}$ with corresponding block domain and correct encoding such that it holds ${\forall t\in\N, \gmaun^t(c_1,z_1,q_1) = \Gamma\bigl(\gmadeux^{kt}(c_2,z_2,q_2))}$.
\end{definition}

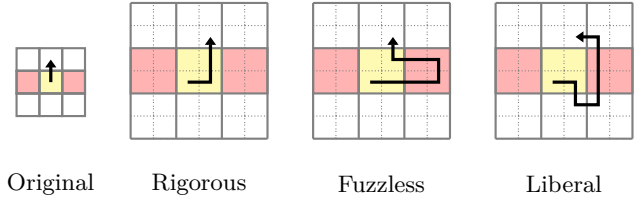
\begin{figure}
  \centering
  \begin{tikzpicture}[scale = .6,
    rrr/.style={fill = red, draw = gray, fill opacity=.3, draw opacity=1, text opacity = 1, text = black, thick},
    yyy/.style={fill = yellow, draw = gray, fill opacity=.3, draw opacity=1, text opacity = 1, text = black, thick},
    ttt/.style={fill = white, draw = gray, fill opacity=0, draw opacity=1, text opacity = 1, text = black, thin, dotted},
    uuu/.style={fill = white, draw = gray, fill opacity=0, draw opacity=1, text opacity = 1, text = black, thick},
    bbb/.style={fill = blue, draw = gray, fill opacity=.3, draw opacity=1, text opacity = 1, text = black, thick},
    arr/.style={->, >=stealth, ultra thick, black},
    >={Triangle[]},
    arr/.style={-{Triangle[length=3pt,width=4pt]},black,very thick},
    dot/.style={black, line width=1pt, line cap=round, dash pattern=on 0pt off 4}
    ]
    \begin{scope}[shift={(-6,0)}]
      \draw[rrr] (-.5,0) rectangle +(.5,.5);
      \draw[rrr] (.5,0) rectangle +(.5,.5);
      \draw[yyy] (0,0) rectangle +(.5,.5);
      \foreach \j in {-.5,0,...,.5}{
        \foreach \i in {-.5,0,...,.5}{
          \draw[uuu] (\j,\i) rectangle +(.5,.5);
        }}
        \draw[arr] (.25,.25) -- (.25,.75);
      \draw (.25,-2) node {\small Original};
    \end{scope}
    \begin{scope}[shift={(-3,0)}]
      \draw[rrr] (-1,0) rectangle +(1,1);
      \draw[rrr] (1,0) rectangle +(1,1);
      \draw[yyy] (0,0) rectangle +(1,1);
      \foreach \j in {-1,-.5,...,1.5}{
        \foreach \i in {-1,-.5,...,1.5}{
          \draw[ttt] (\j,\i) rectangle +(.5,.5);
        }}
      \draw[uuu] (-1,-1) -- +(3,0);
      \draw[uuu] (-1,2) -- +(3,0);
      \foreach \j in {-1,0,...,2}{
        \draw[uuu] (\j,-1) -- +(0,3);
      }
      \draw[arr] (.25,.25) -- ++(.5,0) -- ++(0,.5) -- ++(0,.5);
      \draw (.5,-2) node {\small Rigorous};
    \end{scope}
    \begin{scope}[shift={(1,0)}]
      \draw[rrr] (-1,0) rectangle +(1,1);
      \draw[rrr] (1,0) rectangle +(1,1);
      \draw[yyy] (0,0) rectangle +(1,1);
      \foreach \j in {-1,-.5,...,1.5}{
        \foreach \i in {-1,-.5,...,1.5}{
          \draw[ttt] (\j,\i) rectangle +(.5,.5);
        }}
      \draw[uuu] (-1,-1) -- +(3,0);
      \draw[uuu] (-1,2) -- +(3,0);
      \foreach \j in {-1,0,...,2}{
        \draw[uuu] (\j,-1) -- +(0,3);
      }
      \draw[arr] (.25,.25) -- ++(1.5,0) -- ++(0,.5) -- ++(-1,0)-- ++(0,.5);
      \draw (.5,-2) node {\small Fuzzless};
    \end{scope}
    \begin{scope}[shift={(5,0)}]
      \draw[rrr] (-1,0) rectangle +(1,1);
      \draw[rrr] (1,0) rectangle +(1,1);
      \draw[yyy] (0,0) rectangle +(1,1);
      \foreach \j in {-1,-.5,...,1.5}{
        \foreach \i in {-1,-.5,...,1.5}{
          \draw[ttt] (\j,\i) rectangle +(.5,.5);
        }}
      \draw[uuu] (-1,-1) -- +(3,0);
      \draw[uuu] (-1,2) -- +(3,0);
      \foreach \j in {-1,0,...,2}{
        \draw[uuu] (\j,-1) -- +(0,3);
      }
      \draw[arr] (.25,.25) -- ++(.5,0) -- ++(0,-.5) -- ++(.5,0) -- ++(0,1.5) -- ++(-.5,0);
      \draw (.5,-2) node {\small Liberal};
    \end{scope}
    
  \end{tikzpicture}
  \caption{\label{fig:simtypes}Differences in allowed head movements in rigorous, fuzzless and liberal simulations with ${2\times 2}$ blocks. The colors have the following meaning: in red the positions or blocks which are not empty initially, in yellow the positions or blocks coding a non-$\blanc$ letter during the orbit; in white the positions or blocks coding $\blanc$; in black the movement of the head.}
\end{figure}

Liberal simulations can have fuzz and make non-rigorous head movements.
Note that fuzzless simulations are equivalent to simulations where the headless block decoding map is such that ${\alpha(u)=\blanc\iff u=\blanc^{\brect{b}}}$, \textit{i.e.} that the only block coding $\blanc$ is $\blanc^{\brect{b}}$.
Note also that a rigorous simulation is necessarily fuzzless because the head of the simulator has no opportunity, even at intermediate time steps, to visit blocks not corresponding to the domain of the simulated configuration.
The power of fuzzless simulations compared to rigorous ones is to allow the head to go back to blocks that were previously visited.
Thus the head can potentially retrieve information from non adjacent blocks that were written a long time ago.
With rigorous simulation on the contrary, the head has only access to adjacent blocks during a simulation cycle.

\newcommand\simlib{\ensuremath{\leq}}
\newcommand\simnofuzz{\ensuremath{\leq_{FL}}}
\newcommand\simrig{\ensuremath{\leq_{R}}}
We denote by \simlib{} the liberal simulation, by \simnofuzz{} the fuzzless simulation and by \simrig{} the rigorous simulation.
\newcommand\univsim[3]{\mathcal{U}_{#1}^{#2}(#3)}
\newcommand\univsimdim[2]{\mathcal{U}_{#1}^{#2}}
Among many properties of these simulation relations, we are particularly interested in universality: the capacity of a single turedo to simulate a whole set of turedos.
Given a dimension $d$ and a radius $r$, we denote by ${\univsim{d}{\simlib}{r}}$ the set of turedos ${T\in\turdim{d}}$ such that for any ${T'\in\tur{d}{r}}$ it holds ${T'\simlib T}$.
We denote by ${\univsimdim{d}{\simlib}}$ the set of turedos ${T\in\turdim{d}}$ such that for any ${T'\in\turdim{d}}$ it holds ${T'\simlib T}$.
We use similar notations for simulation relations $\simrig$ and $\simnofuzz$.

\section{Separation Results}
\label{sec:hier}

\subsection{No Fuzz, no Fun}
\label{sec:nofuzzneg}

The fuzzless condition gives much importance to larger radii, simply because a turedo's head surrounded by blocks coding $\blanc$ cannot read information far away without moving inside these blocks and thus creating fuzz.
The following theorem exploits this obvious limitation to show two immediate consequences: first, there is a radius hierarchy (new behaviors appear at radius $r+1$ that cannot be simulated at radius $r$) and thus no general fuzzless universality; second, universality is impossible even at radius $2$: any turedo (whatever its radius) will fail to simulate some radius-2 turedo.
The dimension plays no role in these results.

\begin{theorem}\label{thm:nofuzznofun} For any $d\geq 2$ and $r\geq 1$, we have the following:
  \begin{itemize}
  \item there is ${T_{r+1}\in\tur{d}{r+1}}$ such that
    for all ${T_r\in\tur{d}{r}}$, ${T_{r+1}\not\simnofuzz T_r}$ ; in
    particular, ${\univsimdim{d}{\simnofuzz}=\emptyset}$.
  \item for any ${T_r\in\tur{d}{r}}$ there exists ${T_2\in\tur{d}{2}}$ such that ${T_2\not\simnofuzz T_r}$ ; in particular, ${\univsim{d}{\simnofuzz}{r}=\emptyset}$ for any ${r\geq 2}$.
  \end{itemize}
\end{theorem}
\begin{proof}
  We prove the result for ${d=2}$, the argument can be generalized to higher dimension straightforwardly (by completing 2D configurations by $\blanc$ everywhere else).

  For the first item, simply consider a turedo ${T_{r+1}\in\tur{2}{r+1}}$ that has the following behavior when the head in position ${(0,0)}$ has only $\blanc$ letters at the north and at the south of its current position: read the letters at positions ${(r+1,0)}$ and ${(0,r+1)}$ and move to the north is they are equal and to the south otherwise.
  Consider any turedo ${T_r\in\tur{2}{r}}$ and any block size ${b\in\N_+^2}$. To simulate $T_{r+1}$ fuzzlessly, ${T_r}$ has to move either to block ${b\otimes(0,1)}$ or block ${b\otimes(0,-1)}$ depending on blocks ${b\otimes(r+1,0)}$ and ${b\otimes (0,r+1)}$ without entering into any other neighboring block: this is impossible, because with radius $r$ turedo $T_r$ can't have any information about either ${b\otimes(r+1,0)}$ or ${b\otimes(0,r+1)}$ before making a decisive move (by entering inside either ${b\otimes(0,1)}$ or ${b\otimes(0,-1)}$) so it will fail to correctly simulate the orbit of at least one seed.

  The second item can be proved similarly, using a pumping trick on the alphabet: for any fixed $T_r\in\tur{2}{r}$ with alphabet of cardinal $k$, choose ${T_2\in\tur{2}{2}}$ with an alphabet strictly larger than ${k^{r^2}}$ so that at least on dimension of the block size $b$ of any potential simulation of ${T_2}$ by $T_r$ has to be at least ${r+1}$ (otherwise there is simply no way to code all letters of $T_2$ on different blocks of size $b$).
  Then, choosing $T_2$ to have the same behavior as $T_{r+1}$ above, we get the same contradiction: there is a direction of $b$, let's say the vertical one, which overwhelms the radius $r$ of $T_r$ so $T_r$ won't be able to read the content of block ${b\otimes (0,2)}$ before making a decisive move and will therefore fail to correctly simulate at least one orbit.
\end{proof}

\subsection{Dimension $2$ and Radius $1$: the Jordan Curve Burden}
\label{sec:1dr1neg}

A turedo's head in dimension 2 always moves drawing a $4$-connected path.
When the turedo has radius $1$, it has no way to read information across such a path (while it could with a larger radius).
Therefore head movements for turedos of radius $1$ turns into potential information barriers.
The precise way in which this simple observation affects the simulation power of such turedos depends on the type of simulation considered.

\begin{figure}
  \centering
  \begin{tikzpicture}[scale = .5,
    rrr/.style={fill = red, draw = gray, fill opacity=.3, draw opacity=1, text opacity = 1, text = black, thick},
    yyy/.style={fill = yellow, draw = gray, fill opacity=.3, draw opacity=1, text opacity = 1, text = black, thick},
    ttt/.style={fill = yellow, draw = gray, fill opacity=0, draw opacity=1, text opacity = 1, text = black, dotted},
    uuu/.style={fill = white, draw = gray, fill opacity=0, draw opacity=1, text opacity = 1, text = black, thick},
    bbb/.style={fill = blue, draw = gray, fill opacity=.3, draw opacity=1, text opacity = 1, text = black, thick},
    arr/.style={->, >=stealth, ultra thick, gray},
    >={Triangle[]},
    arr/.style={-{Triangle[length=4pt,width=6pt]},gray},
    dot/.style={black, line width=1pt, line cap=round, dash pattern=on 0pt off 4}
    ]
    \begin{scope}[shift={(-7,1)}]
      \draw[yyy, opacity=1] (-1,0) rectangle +(1,1) node[pos=.5] {$.$};
      \draw[yyy] (0,0) rectangle +(1,1) node[pos=.5] {$a$};
      \draw[rrr] (1,0) rectangle +(1,1) node[pos=.5] {$a$};
      \draw [arr] (0,.5) -- (0.13,.5);      
      \draw [arr] (.5,1) -- (.5,1.13);
    \end{scope}
    \draw[yyy, opacity=1] (-1,0) rectangle +(1,1) node[pos=.5] {$.$};
    \draw[yyy] (-1,1) rectangle +(1,1) node[pos=.5] {$.$};
    \draw[yyy] (0,1) rectangle +(1,1) node[pos=.5] {$.$};
    \draw[yyy] (1,1) rectangle +(1,1) node[pos=.5] {$.$};
    \draw[yyy] (1,0) rectangle +(1,1) node[pos=.5] {$.$};
    \draw[yyy] (2,0) rectangle +(1,1) node[pos=.5] {$X$};
    \draw[yyy] (2,1) rectangle +(1,1) node[pos=.5] {$.$};
    \draw[yyy] (2,2) rectangle +(1,1) node[pos=.5] {$.$};
    \draw [arr] (-.5,1) -- (-.5,1.13);    
    \draw [arr] (0,1.5) -- (.13,1.5);    
    \draw [arr] (1,1.5) -- (1.13,1.5);    
    \draw [arr] (1.5,1) -- (1.5,0.87);    
    \draw [arr] (2,.5) -- (2.13,.5);
    \draw [arr] (2.5,1) -- (2.5,1.13);    
    \draw [arr] (2.5,2) -- (2.5,2.13);    
    \draw [arr] (2.5,3) -- (2.5,3.13);

    \draw[rrr] (3,0) rectangle +(1,1) node[pos=.5] {$a_1$};
    \draw[rrr] (4,1) rectangle +(1,1) node[pos=.5] {$a_2$};
    \draw[rrr] (3,2) rectangle +(1,1) node[pos=.5] {$a_3$};

    \foreach \j in {-2,-1,...,5}{
      \draw[ttt] (\j,0) -- +(0,3);
    }
    \draw[ttt] (-3,1) -- +(9,0);
    \draw[ttt] (-3,2) -- +(9,0);
    \draw[uuu] (0,0) rectangle +(3,3);
    \draw[uuu] (-3,0) rectangle +(3,3);
    \draw[uuu] (3,0) rectangle +(3,3);
      
  \end{tikzpicture}  
  \caption{Example of a copy-and-move operation (on the left) and its rigorous simulation by a turedo of radius $1$ with ${3\times 3}$ block size (on the right). The color convention is as follows: in red the letters present in the seed, in dark yellow the initial position of the head, and in light yellow, the positions visited by the head during the orbit. The only position on the south side of the middle block that can depend on $a$ is the lower right corner, marked with a $X$.\label{fig:copyandmoverigorous}}
\end{figure}

Let us first consider rigorous simulations.
Any turedo $T$ (whatever its radius) can obviously do the following elementary copy-and-move operation (see Figure~\ref{fig:copyandmoverigorous}):
\begin{itemize}
\item move to the right to some position $z$;
\item read the letter $a$ present at position ${z+(1,0)}$;
\item then move to position ${z+(0,1)}$ and leave behind letter $a$ at position $z$.
\end{itemize}
In particular, if the head continues its way and later arrives at position ${z-(0,1)}$ from the south, it can read the information $a$ copied at position $z$.

However, if we suppose that some ${T_1\in\tur{2}{1}}$ simulates $T$ under rigorous simulations with block size $b$, the movement of its head inside block ${b\otimes z}$ to simulate the above copy-and-move step must be the following (see Figure~\ref{fig:copyandmoverigorous}):
\begin{itemize}
\item coming from the left side of the block, it draws some path inside it until it reaches the right border (if not it cannot read any information from the adjacent block to the right);
\item then it must move north, otherwise it would be trapped in the south part of the blocks by the $4$-connected path drawn so far that connects the left and right sides of the block;
\item it must finally escape through the north side.
\end{itemize}
This head movement is such that at most $1$ letter of $T_1$ is written on the south side of block ${b\otimes z}$ after having had the opportunity to read some information from the adjacent block that encodes letter $a$.
In particular, if $T_1$ has smaller alphabet than $T$ this is not enough to completely encode $a$ on the south side of block ${b\otimes z}$.
So, this is a limitation that has to be dealt with if later in the simulation the head arrives from the south and has to read from the south side of block ${b\otimes z}$. 
A single copy-and-move is not enough to get a contradiction because the simulation of $T_1$ could be organized so as to transport the complete information about $a$ along the way and have it on hands already when arriving at the south of block ${b\otimes z}$.
However, by repeating such copy-and-move steps, one can saturate the simulator and show the following theorem that states that there is no universal turedo of radius $1$ among turedos of radius $1$ for rigorous simulations.

\newcommand\seed{\sigma}

\begin{figure}
  \centering
  \begin{tikzpicture}[scale = .6,
    rrr/.style={fill = red, draw = gray, fill opacity=.3, draw opacity=1, text opacity = 1, text = black, thick},
    yyy/.style={fill = yellow, draw = gray, fill opacity=.3, draw opacity=1, text opacity = 1, text = black, thick},
    ttt/.style={fill = yellow, draw = gray, fill opacity=0, draw opacity=0, text opacity = 1, text = black, thick},
    bbb/.style={fill = blue, draw = gray, fill opacity=.3, draw opacity=1, text opacity = 1, text = black, thick},
    arr/.style={->, >=stealth, ultra thick, gray},
    >={Triangle[]},
    arr/.style={-{Triangle[length=4pt,width=6pt]},gray},
    dot/.style={black, line width=1pt, line cap=round, dash pattern=on 0pt off 4}
    ]

    \draw[rrr] (0,0) rectangle +(1,1) node[pos=.5] {\footnotesize $a_0$};
    \draw[rrr] (3,0) rectangle +(1,1) node[pos=.5] {\footnotesize $a_1$};
    \draw[rrr] (6,0) rectangle +(1,1) node[pos=.5] {\footnotesize $a_2$};
    \draw[rrr] (13,0) rectangle +(1,1) node[pos=.5] {\footnotesize $a_n$};

    \draw[rrr] (5,-3) rectangle +(1,1) node[pos=.5] {\footnotesize $a$};

    \draw[bbb] (16,0) rectangle +(1,1) node[pos=.5] {\tiny $\downarrow$};
    \draw[bbb] (16,-1) rectangle +(1,1) node[pos=.5] {\tiny $\downarrow$};
    \draw[bbb] (16,-2) rectangle +(1,1) node[pos=.5] {\tiny $\leftarrow$};
    \draw[bbb] (4,-2) rectangle +(1,1) node[pos=.5] {\tiny $\uparrow$};

    \draw[yyy, opacity=1] (-2,0) rectangle +(1,1) node[pos=.5] {$.$};
    \draw[yyy] (-1,1) rectangle +(1,1) node[pos=.5] {$.$};
    \draw[yyy] (0,1) rectangle +(1,1) node[pos=.5] {$.$};
    \draw[yyy] (1,1) rectangle +(1,1) node[pos=.5] {$.$};
    \draw[yyy] (1,0) rectangle +(1,1) node[pos=.5] {$.$};
    \draw[yyy] (-1,0) rectangle +(1,1) node[pos=.5] {\footnotesize $a_0$};
    \draw[yyy] (2,0) rectangle +(1,1) node[pos=.5] {\footnotesize $a_1$};
    \draw[yyy] (4,0) rectangle +(1,1) node[pos=.5] {$.$};
    \draw[yyy] (2,1) rectangle +(1,1) node[pos=.5] {$.$};
    \draw[yyy] (3,1) rectangle +(1,1) node[pos=.5] {$.$};
    \draw[yyy] (4,1) rectangle +(1,1) node[pos=.5] {$.$};
    \draw[yyy] (5,0) rectangle +(1,1) node[pos=.5] {\footnotesize $a_2$};
    \draw[yyy] (5,1) rectangle +(1,1) node[pos=.5] {$.$};
    \draw[yyy] (6,1) rectangle +(1,1) node[pos=.5] {$.$};
    \draw[yyy] (7,1) rectangle +(1,1) node[pos=.5] {$.$};
    \draw[yyy] (7,0) rectangle +(1,1) node[pos=.5] {$.$};
    \draw[yyy] (8,0) rectangle +(1,1) node[pos=.5] {\footnotesize $a_3$};
    \draw[ttt] (8.8,0) rectangle +(1,1) node[pos=.5] {$...$};


    \draw[ttt] (10.2,0) rectangle +(1,1) node[pos=.5] {$...$};
    \draw[yyy] (11,0) rectangle +(1,1) node[pos=.5] {$.$};
    \draw[yyy] (12,0) rectangle +(1,1) node[pos=.5] {\footnotesize $a_n$};
    \draw[yyy] (12,1) rectangle +(1,1) node[pos=.5] {$.$};
    \draw[yyy] (13,1) rectangle +(1,1) node[pos=.5] {$.$};
    \draw[yyy] (14,1) rectangle +(1,1) node[pos=.5] {$.$};
    \draw[yyy] (14,0) rectangle +(1,1) node[pos=.5] {$.$};
    \draw[yyy] (15,0) rectangle +(1,1) node[pos=.5] {$.$};
    \draw[yyy] (15,-1) rectangle +(1,1) node[pos=.5] {$.$};

    \foreach \i in {15,...,11}{
      \draw[yyy] (\i,-2) rectangle +(1,1) node[pos=.5] {$.$};
    }
    \foreach \i in {8,...,6}{
      \draw[yyy] (\i,-2) rectangle +(1,1) node[pos=.5] {$.$};
    }

    \draw[yyy] (5,-2) rectangle +(1,1) node[pos=.5] {\footnotesize $a$};
    \draw[yyy, fill opacity=1, fill=deepsaffron] (5,-1) rectangle +(1,1) node[pos=.5] {\footnotesize $=$?};

    \draw [arr] (-1,.5) -- (-.87,.5);
    \draw [arr] (2,.5) -- (2.13,.5);
    \draw [arr] (5,.5) -- (5.13,.5);
    \draw [arr] (8,.5) -- (8.13,.5);
    \draw [arr] (12,.5) -- (12.13,.5);
    \draw [arr] (15,.5) -- (15.13,.5);

    \draw [arr] (0,1.5) -- (.13,1.5);
    \draw [arr] (1,1.5) -- (1.13,1.5);
    \draw [arr] (3,1.5) -- (3.13,1.5);
    \draw [arr] (4,1.5) -- (4.13,1.5);
    \draw [arr] (6,1.5) -- (6.13,1.5);
    \draw [arr] (7,1.5) -- (7.13,1.5);
    \draw [arr] (13,1.5) -- (13.13,1.5);
    \draw [arr] (14,1.5) -- (14.13,1.5);

    \draw [arr] (-.5,1) -- (-.5,1.13);
    \draw [arr] (2.5,1) -- (2.5,1.13);
    \draw [arr] (5.5,1) -- (5.5,1.13);
    \draw [arr] (8.5,1) -- (8.5,1.13);
    \draw [arr] (12.5,1) -- (12.5,1.13);

    \draw [arr] (1.5,1) -- (1.5,.87);
    \draw [arr] (4.5,1) -- (4.5,.87);
    \draw [arr] (7.5,1) -- (7.5,.87);
    \draw [arr] (11.5,1) -- (11.5,.87);
    \draw [arr] (14.5,1) -- (14.5,.87);
    \draw [arr] (15.5,0) -- (15.5,-.13);
    \draw [arr] (15.5,-1) -- (15.5,-1.13);

    \foreach \i in {15,...,11}{
      \draw [arr] (\i,-1.5) -- (\i-.13,-1.5);
    }
    \foreach \i in {9,...,6}{
      \draw [arr] (\i,-1.5) -- (\i-.13,-1.5);
    }


    \draw[dot] (9.25,-1.5) -- (10.75,-1.5);


    \draw [arr] (5.5,-1) -- (5.5,-.87);

    \draw [arr] (6,-.5) -- (6.13,-.5);
    \draw [arr] (5,-.5) -- (4.87,-.5);
    \draw[ttt] (6,-1) rectangle +(1,1) node[pos=.5] {\color{gray}\textbf{\tiny YES}};
    \draw[ttt] (4,-1) rectangle +(1,1) node[pos=.5] {\color{gray}\textbf{\tiny NO}};

  \end{tikzpicture}
  \caption{Behavior of turedo $T'$ on the seed ${\seed(\vec{a},2,a)}$. The red and blue colors indicate letters present in the seed. The dark yellow color indicates the initial position of the head and light yellow cells represent the path of the head until the last step of the orbit. The orange cell correspond to the position of the head before moving left or right according to the result of the test 
    .\label{fig:nounivradius1}}  
\end{figure}
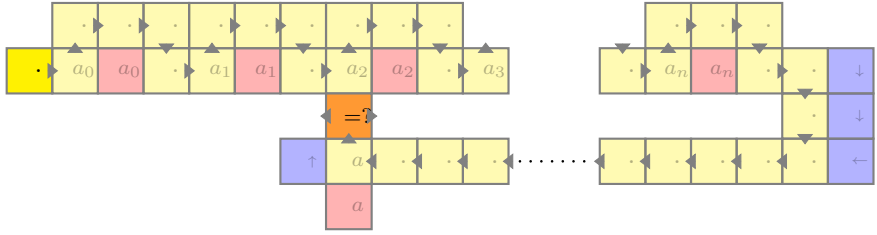

\begin{theorem}\label{theo:noradius1univ}
  For any ${T\in\tur{2}{1}}$ there is ${T'\in\tur{2}{1}}$ such that ${T'\not\simrig T}$.
  In particular ${\univsim{2}{\simrig}{1}\cap\tur{2}{1}=\emptyset}$.
\end{theorem}

\begin{proof}
  Let $Q$ be the state set of $T$, $A$ be the alphabet of $T$ and consider any alphabet $A_+$ with ${m=|A|<|A_+|=m_+}$.
  Then it is straightforward to construct a turedo ${T'\in\tur{2}{1}}$ of alphabet $A'=A_+\cup\{\downarrow,\leftarrow,\uparrow\}$ that has the following behavior (see Figure~\ref{fig:nounivradius1}):
  \begin{itemize}
  \item for any ${n\in\N}$, any ${\vec{a}=(a_0,\ldots, a_n)\in A_+^{n+1}}$, ${a'\in A_+}$ and ${0\leq i\leq n}$, consider the finite seed ${\seed(\vec{a},i,a')}$ with head in position ${(0,0)}$, $a_j$ in position ${(3j+2,0)}$ for ${0\leq j\leq n}$, $a'$ in position ${(3i+1,-3)}$ and ${\downarrow}$ in positions ${(3(n+1)+2,0)}$ and ${(3(n+1)+2,-1)}$, $\leftarrow$ in position ${(3(n+1)+2,-2)}$ and finally $\uparrow$ at position ${(3i,-2)}$;
  \item from such a seed, $T'$ starts a sequence of ${n+1}$ copy-and-move steps that results in having a copy of ${a_j}$ at position ${(3j+1,0)}$ for ${0\leq j\leq n}$; the end of this phase occurs at time step ${5(n+1)}$ and the head reaches position ${(3(n+1),0)}$;
  \item then $T'$ reaches the first ${\downarrow}$ and follows the move indications of arrows (down, down, left), until it reaches the up arrow, and moves from position ${(3i+1,-2)}$ to ${(3i+1,-1)}$;
  \item finally, at position ${(3i+1,-1)}$ it moves right if ${a_i=a'}$ and left otherwise (it can do so because it has copied the value of $a'$ when leaving position ${(3i+1,-2)}$).
  \end{itemize}
  Let's call $t_{n,i}$ the time step at which occurs this final left or right move ($t_{n,i}$ only depends on $i$ and $n$): at time $t_{n,i}$, the head of $T'$ must be either at position ${(3i,-1)}$ or ${(3i+2,-1)}$.
  Thus $T'$ implements on seed ${\seed(\vec{a},i,a')}$ the test of whether ${a_i=a'}$.
  We are going to show that $T$ cannot simulate $T'$ under rigorous simulations.
  Suppose by contradiction that ${T'\simrig T}$ with block size $b$ and time scaling factor $k$.
  Given ${n\in\N}$ and ${0\leq i\leq n}$, denote by ${A_{n}\subseteq\Z^2}$ the set of positions that are on the right side of block ${b\otimes (3(n+1),0)}$ (the block corresponding to the position reached by $T'$ at the end of the copy-and-move sequence as detailed above).
  Denote by ${B_{i,n}\subseteq\Z^2}$ the set of positions that are on the south side of block ${b\otimes(3i+1,0)}$.
  Finally, denote by ${C_{i,n}\subseteq\Z^2}$ the set of positions made of the union of blocks ${(3i+1,-3)}$,  ${(3i,-2)}$, ${(3(n+1)+2,0)}$, ${(3(n+1)+2,-1)}$, ${(3(n+1)+2,-2)}$ (\textit{i.e.} those corresponding to position $a'$ or an arrow ${\{\downarrow,\leftarrow,\uparrow\}}$ in the seed ${\seed(\vec{a},i,a')}$).
  Consider now ${n\in\N}$, ${\vec{a},\vec{c}\in A_+^{n+1}}$ ${a'\in A_+}$ and ${0\leq i\leq n}$, and take any two global states ${g_1}$ and ${g_2}$ of $T$ that correctly simulate the orbits of $T'$ on seed ${\seed(\vec{a},i,a')}$ and ${\seed(\vec{c},i,a')}$ respectively and that are identical on $C_{i,n}$
  Considering time step ${t_0(n)=5k(n+1)}$ corresponding to the end of the copy-and-move sequence, if global states ${T^{t_0(n)}(g_1)}$ and ${T^{t_0(n)}(g_2)}$ are identical on domains ${A_n}$ and ${B_{i,n}}$ and have the same head state, then both orbits must make the same final decision to move to the left block or the right block at the final time step ${kt_{n,i}}$, precisely: the head in global state ${T^{kt_{n,i}}(g_1)}$ is in the same block as the head in global state ${T^{kt_{n,i}}(g_2)}$ (and it must be either ${b\otimes (3i,-1)}$ or ${b\otimes (3i+2,-1)}$).
  Indeed, by the property of rigorous simulations and the behavior of $T'$, the only positions with content written before ${t_0(n)}$ that the head of $T$ can possibly read between time step ${t_0(n)}$ and ${kt_{n,i}}$ are positions in ${A_n\cup B_{i,n}\cup C_{i,n}}$, so the orbit starting from step $t_0(n)$ is completely determined by the content of the configuration in that domain and the internal state of $T$ at time $t_0(n)$.
  \begin{claim}
    There must exist ${n\in\N}$, ${0\leq i\leq n}$, ${\vec{a}\in A_+^{n+1}}$, ${a'\in A_+}$ and ${\vec{c}\in A_+^{n+1}}$ with ${a_j=c_j}$ for all ${j<i}$ and ${a_i\neq c_i}$, and two global states ${g_1}$ and ${g_2}$ of $T$ that correctly simulate seeds ${\seed(\vec{a},i,a')}$ and ${\seed(\vec{c},i,a')}$ respectively, and also such that ${T^{t_0(n)}(g_1)}$ and ${T^{t_0(n)}(g_2)}$ have same head state and are identical on domain ${A_n\cup B_{i,n}\cup C_{i,n}}$.
  \end{claim}
  \begin{proof}[Proof of the claim]
    In this proof, we fix for each ${a\in A_+}$ a unique block of ${A^{\brect{b}}}$ that encodes it, and for any seed of $T'$ we only consider a unique global state of $T$ that simulates it.
    First, there are only a bounded number (bound in $n$) of possible content of a configuration on domain $A_n$ and state of $T$, so for each $n$ there must exist ${u\in A^{A_n}}$ and ${q\in Q}$, a set ${X_n\subseteq A_+^n}$ of size ${\Omega(m_+^n)}$ such that for each ${0\leq i\leq n}$ and each $\vec{a}\in X_n$, the corresponding global state $g$ of $T$ simulating $T'$ on seed ${\seed(\vec{a},i,n)}$, is such that ${T^{t_0(n)}(g)}$ is equal to $u$ on domain $A_n$ and with head state $q$.
    
    Second, we claim that for large enough $n$ there must be some ${i}$ and a prefix ${a_0,\ldots, a_{i-1}\in A_+^i}$ such that there are at least ${m+1}$ choices of ${a_i\in A_+}$ such that ${a_0,\ldots, a_i}$ can be completed into an element ${\vec{a}\in X_n}$.
    Indeed, otherwise we would have ${|X|\leq m^n}$ which would contradict the fact that ${|X|\in\Omega(m_+^n)}$ for large enough $n$ since ${m<m_+}$.

    Now consider the set of global states that simulates the seeds ${\seed(\vec{a},i,a')}$ where $\vec{a}\in X_n$ are the ${m+1}$ completed vectors from the common prefix ${a_0,\ldots,a_{i-1}}$, and $a'\in A_+$.
    They are identical on the blocks corresponding to the common prefix ${a_0,\ldots, a_{i-1}}$ of the seed they simulate.
    As already said, these global states at step ${t_0(n)}$ are also identical on domain $A_n$ and have same head state.
    Moreover, on domain $B_{i,n}$ and still at step ${t_0(n)}$, they agree because of the common prefix ${a_0,\ldots,a_{i-1}}$, except possibly on the lower-right corner where they can take at most $m$ different values (see Figure~\ref{fig:copyandmoverigorous} and discussion at the beginning of this section).
    We deduce that among the $m+1$ choices for $a_i$, at least $2$ must correspond to global states that completely agree on $B_{i,n}$.
    Denote by $a'$ and $c'$ these two choices and consider $\vec{a}$ and $\vec{c}$ to be the vectors completing the prefixes ${a_0,\ldots,a_{i-1},a'}$ and ${a_0,\ldots, a_{i-1},c'}$ respectively.
    The claim follows by choosing seeds ${\seed(\vec{a},i,a')}$ and ${\seed(\vec{c},i,a')}$.
  \end{proof}
  The theorem follows from the claim by contradiction: as shown above, global states ${g_1}$ and ${g_2}$ force the same behavior of $T$ starting from time $t_0(n)$, but at the same time their orbits should not end up in the same block because they simulate seeds of $T'$ that do not have the same answer to the final equality test.
\end{proof}

We will now establish a strong separation between ${\tur{2}{1}}$ and ${\tur{2}{2}}$ even under liberal simulations.
We first establish a lemma expressing bounds on information leakage between too regions separated by a $4$-connected path.
It is formulated using Kolmogorov complexity.
Recall that the (plain) Kolmogorov complexity of a string ${u\in\{0,1\}^*}$ is the length of the shortest program that outputs $u$, more precisely the length of the shortest ${v\in\{0,1\}^*}$ such that a suitable fixed universal Turing machine outputs $u$ on input $v$ (see \cite{Li_2008}).
\newcommand\kolmo[1]{\ensuremath{K(#1)}}
For any ${X\subseteq\Z^2}$ and any (partial) configuration ${c\in Q^X}$ of finite domain, we denote by ${\kolmo{c}}$ its Kolmogorov complexity, which is the Kolmogorov complexity of the finite binary string $u$
that encodes $c$ as a list of pairs ${(z,c(z))}$ such that ${c(z)\neq\blanc}$ given in lexicographical order.

\begin{lemma}\label{lem:kolmojordan}
  Let $C_0\in\N$ be some constant and ${T\in\tur{2}{1}}$.
  Then there is another constant $C\in\N$ with the following property.
  Consider any $4$-connected path $\rho$ of $\Z^2$ that divides $\Z^2$ in $2$ or more connected components, and any finite global state ${s\in\globstat{T}}$ with head at position ${(0,0)}$, and whose domain ${\dom{s}}$ lies entirely in one of the connected components defined by $\rho$, denoted $A_0$.
  Suppose moreover that for some ${n\in\N}$, the orbit from global state $s$ to global state ${(c,z,q)=\globmap{T}^n(s)}$ is such that the head visits path $\rho$ at most $C_0$ times.
  Then, the restriction of $c$ to the complement of $A_0$ has 'small' kolmogorov complexity: ${\kolmo{c_{|\Z^2\setminus A_0}}\leq C\log(n)}$.
\end{lemma}
\begin{proof}
  We show that ${c_{|\Z^2\setminus A_0}}$ can be computed from the following description ${\mathcal{D}}$:
  \begin{itemize}
  \item the finite list of positions on $\rho$ that are visited by the head during the $n$ first steps of the run starting from $s$;
  \item the list of events corresponding to each such position $z$ given as a triple: time at which the head leaves position $z$, letter written at that step, and move made by the head at that step.
  \end{itemize}
  This description is of size ${O(\log(n))}$ because both positions $z$ and time steps occurring in the above lists are bounded by $n$ by definition (recall that the head is initially at ${(0,0)}$). 

  Because $T$ is of radius $1$ and $\rho$ is $4$-connected, each time the head of the turedo is neither in $A_0$ nor on $\rho$, the local transition does not depend on the current configuration on domain $A_0$.
  A Turing machine can therefore compute ${c_{|\Z^2\setminus A_0}}$ from this description by maintaining the following partial information step by step:
  \begin{itemize}
  \item the current configuration restricted to domain ${\Z^2\setminus A_0}$,
  \item the partial information on the head position $z$: the exact position if ${z\not\in A_0}$ or the state ``undefined'' else.
  \end{itemize}
  This information is straightforward at the initial step since ${\dom{s}\subseteq A_0}$ so the head is in $A_0$ and the configuration is $\blanc$ everywhere outside $A_0$.
  The partial information at step $n$ is enough to give ${c_{|\Z^2\setminus A_0}}$ and it is updated from one step $i$ to the next $i+$ as follows:
  \begin{itemize}
  \item if the partial information on the head at step $i$ is undefined and time step $i+2$ does not appear in the lists of $\mathcal{D}$, then don't change the partial information (the head is in $A_0$ and won't move to $\rho$ at step $i+1$);
  \item if the head information is undefined but step $i+2$ appears in $\mathcal{D}$, then updates the head position to the position on $\rho$ that corresponds to the item stamped by time steps $i+2$ in $\mathcal{D}$;
  \item if the head position is on $\rho$ then some item in the list of $\mathcal{D}$ must be stamped by time step $i+1$ and gives all the information to update both the head position and the configuration on $\rho$;
  \item finally if the head position is neither in $A_0$ nor on $\rho$, then the knowledge of the current configuration restricted to domain ${\Z^2\setminus A_0}$ is enough to update the partial information (position and partial configuration).
  \end{itemize}
\end{proof}

Note that if the seed $s$ has 'large' kolmogorov complexity, for instance $\Omega(n)$, then $n$ steps are far from enough to transmit all the information about $s$ to another connected component under the hypothesis of the above lemma.
The power of this lemma lies in the fact that constant $C$ does not depend on the path $\rho$ nor on the seed $s$.
In particular, one can choose $\rho$ depending on $s$ to apply the lemma.
Turedos of radius $2$ can overcome the limitation of Lemma~\ref{lem:kolmojordan} because they can transmit information over a path without writing on it.
It turns out that this is enough to separate $\tur{2}{2}$ from $\tur{2}{1}$ even under liberal simulations.

\begin{theorem}\label{thm:radius2kolmojordan}
  There is ${T_2\in\tur{2}{2}}$ such that for any ${T_1\in\tur{2}{1}}$: ${T_2\not\simlib T_1}$.
\end{theorem}

\begin{figure}
\centering
\begin{tikzpicture}[scale = .25,
rrr/.style={fill = red, draw = gray, fill opacity=.3, draw opacity=1, text opacity = 1, text = black, thick},
yyy/.style={fill = yellow, draw = gray, fill opacity=.3, draw opacity=0, text opacity = 1, text = black, thick},
ttt/.style={fill = yellow, draw = gray, fill opacity=0, draw opacity=0, text opacity = 1, text = black, thick, font = \tiny},
ddd/.style={fill = deepsaffron, draw = gray, fill opacity=.7, draw opacity=0, text opacity = 1, text = black, thick, font = \tiny},
gg/.style={gray, thick},
arr/.style={->, >=stealth, very thick, black, rounded corners=.5ex, opacity = .2},
]

\def\x{12} 
\def\y{11} 
\def\z{10} 

\draw[rrr] (-1,0) rectangle +(1,\x) ;

\draw[yyy] (0,0) rectangle +(\x,\x) ;
\draw[yyy, fill opacity=1] (0,0) rectangle +(1,1);
\draw[yyy] (\x,0) rectangle +(1,1);
\draw[yyy] (\x+1,0) rectangle +(\x,\x) ;

\draw[ddd] (\x,2) -- (\x+1,2) -- (\x+1,\x) -- (2*\x+1,\x) -- (2*\x+1,0) -- (2*\x+2,0) -- (2*\x+2,\x+1) -- (\x,\x+1) -- cycle;

\draw[ddd, fill opacity =1] (\x,1) rectangle +(1,1) node[pos=.5] {};
\draw[ddd, fill opacity =1] (\x,\x) rectangle +(1,1) node[pos=.5] {};

\foreach \i in {0,1,...,\x}{\draw[ttt] (\i-1,0) rectangle +(1,\x) node[pos=.5] { $u$};}
\foreach \i in {1,...,\x}{\draw[ttt] (\x+\i,0) rectangle +(1,\x) node[pos=.5] { $u$};}

\draw [decorate,decoration={brace,amplitude=10pt},xshift=-6pt,yshift=0pt, thick]
(-1,0) -- (-1,\x) node [black,midway,xshift=-0.6cm] 
{\scriptsize $n$};

\draw [decorate,decoration={brace,mirror,amplitude=10pt},xshift=0pt,yshift=-6pt, thick]
(0,-.1) -- (\x,-.1) node [black,midway,below,yshift=-9pt] 
{\scriptsize $n$};

\begin{scope}[shift={(\x+1,0)}]
\draw [decorate,decoration={brace,mirror,amplitude=10pt},xshift=0pt,yshift=-6pt, thick]
(0,-.1) -- (\x,-.1) node [black,midway,below,yshift=-9pt] 
{\scriptsize $n$};
\end{scope}

\foreach \i in {1,3,...,\z}{
\draw[arr] (\i-.5,\x-1.5) -- (\i-.5,\x-.5) -- (\i+.5,\x-.5) -- (\i+.5,\x-1.5);
\draw[arr] (\i-.5+\x+1,\x-1.5) -- (\i-.5+\x+1,\x-.5) -- (\i+.5+\x+1,\x-.5) -- (\i+.5+\x+1,\x-1.5);
\draw[arr] (\i+.5,1.5) -- (\i+.5,.5) -- (\i+1.5,.5) -- (\i+1.5,1.5);
\draw[arr] (\i+.5+\x+1,1.5) -- (\i+.5+\x+1,.5) -- (\i+1.5+\x+1,.5) -- (\i+1.5+\x+1,1.5);
}
\draw[arr] (\x-1.5,\x-1.5) -- (\x-1.5,\x-.5) -- (\x-.5,\x-.5) -- (\x-.5,\x-1.5);
\draw[arr] (\x-1.5+\x+1,\x-1.5) -- (\x-1.5+\x+1,\x-.5) -- (\x-.5+\x+1,\x-.5) -- (\x-.5+\x+1,\x-1.5);
\draw[arr] (\x-.5,1.5) -- (\x-.5,.5) -- (\x+1.5,.5) -- (\x+1.5,1.5); 
\draw[arr] (\x-.5+\x+1,1.5) -- (\x-.5+\x+1,.5) -- (\x+.5+\x+1,.5) -- (\x+.5+\x+1,1.5);
\draw[arr] (\x+.5+\x+1,\x-.5) -- (\x+.5+\x+1,\x+.5) -- (\x-.5+\x+1,\x+.5);
\draw[arr] (\x+.5,.5*\x+1) -- (\x+.5,.5*\x-1);

\draw[gg] (0,0) -- (2*\x+2,0) -- (2*\x+2,\x+1) -- (\x,\x+1) -- (\x,\x) -- (0,\x) -- cycle;
\foreach \i in {1,3,...,\x}{
\draw[gg] (\i,0) -- (\i,\x-1);
\draw[gg] (\i+1,\x) -- (\i+1,1);
}

\draw[gg] (\x,1) -- (\x+1,1) -- (\x+1,\x);

\foreach \i in {1,3,...,\x}{
\draw[gg] (\x+1+\i,0) -- (\x+1+\i,\x-1);
\draw[gg] (\x+1+\i+1,\x) -- (\x+1+\i+1,1);
}

\draw[gg] (\x+1,\x) -- (2*\x+1,\x);

\draw[ddd, fill opacity =0] (\x,1) rectangle +(1,1) node[pos=.5] { $z_S$};
\draw[ddd, fill opacity =0] (\x,\x) rectangle +(1,1) node[pos=.5] { $z_N$};

\end{tikzpicture}
\caption{\label{fig:tur2kolmo}Orbit of turedo $T_2$ starting from a seed $u$ of length $n$ (in red) with the head initially in the position shown in dark yellow. The last part of the orbit is shown in orange.}
\end{figure}
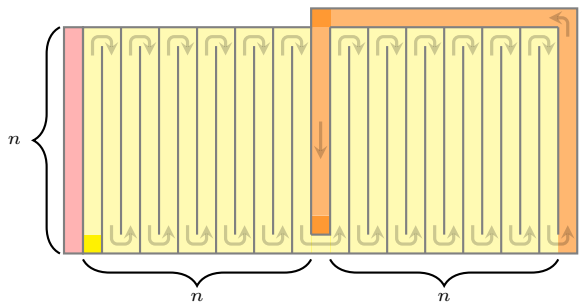

\begin{proof}
  Let's consider the turedo ${T_2\in\tur{2}{2}}$ that behaves as follows on a seed made of a vertical word $u$ of length $n$ (see Figure~\ref{fig:tur2kolmo}):
  \begin{itemize}
  \item it copies $u$ to the right by making zigzags and does this $n$ times (it implements a unary counter initialized to the length of $u$ while making copies);
  \item it then moves one cell to the right without making copies (and thus leaving an empty column above);
  \item it then do again $n$ copies of $u$ by zigzag while moving to the right (note that the first copy can be done because $T$ has radius $2$);
  \item at the end of the last copy it goes around the last bloc of $n$ copies by the north side until it encounters the empty column and then goes down into it until it is blocked.
  \end{itemize}
  Denote by $z_N$ and $z_S$ the northmost and southmost positions of the last sequence of $n$ south moves of the head (see Figure~\ref{fig:tur2kolmo}), and by ${t_N}$ and ${t_S}$ the respective time steps at which the head is at position $z_N$ and $z_S$. Note that $t_S$ is ${O(n^2)}$ and it is the final step of the orbit considered here.

  Now suppose by contradiction that there is some ${T_1\in\tur{2}{1}}$ such that ${T_2\simlib T_1}$.
  Denote by $k$ the time rescaling factor and $b$ the block size involved in this simulation.
  Let's suppose that $n$ is large enough (to be precised later) and that $u$ has large kolmogorov complexity, let's say ${\Omega(n)}$.
  Consider a global initial state $s_1$ for $T_1$ from which starts a correct simulation of the run of $T_2$.
  When the simulation reaches step ${kt_N}$, the head of $T_1$ is inside bloc ${b\otimes z_N}$ and there must be a finite $4$-connected path ${p_1,\ldots,p_m}$ of empty positions from this position to some position inside bloc ${b\otimes z_S}$ because $T_2$ has to simulate the state changes made by $T_1$ between steps $t_N$ and $t_S$ along the vertical segment of positions from ${z_N}$ to ${z_S}$.
  Let $\rho$ denote the infinite path that extends ${p_1,\ldots,p_m}$ infinitely to the north from $p_1$ and infinitely to the south from $p_m$.

  We claim that there is a bound $C_0$ depending only on $T_1$, $b$ and $k$, but not on $n$ and neither on $u$, such that the run of $T_1$ starting from global state $s_1$ until time step ${kt_N}$ crosses at most $C_0$ times path $\rho$.
  First, by choice of $p_1,\cdots, p_m$, such crossings can only happen at positions of $\rho$ that are either at the north of $p_1$ or at the south of $p_m$.
  The simulation is liberal, so the head of $T_1$ has some freedom of move but it must always remain at a bounded distance from the block corresponding to the simulated head position of $T_2$ during intermediate steps, precisely: if the head of $T_2$ is at position $z$ at time step $t$, then the head of $T_1$ must be inside block ${b\otimes z}$ at time step ${kt}$ and therefore at distance at most $k$ of block ${b\otimes z}$ during time steps between ${kt}$ and ${k(t+1)}$.
  A position is therefore potentially reachable by $T_1$ before time step ${kt_N}$ only if it is at distance at most $k$ from a block ${b\otimes z}$ such that position $z$ in the run of $T_2$ is visited before time step $t_N$.
  The key observation is that in the run of $T_2$, there are only finitely many positions that are visited before time $t_N$ and at distance less than $k$ from either ${z_N}$ or any position at the north of it, or from ${z_S}$ or any position at the south of it.
  From this we deduce that ${\rho\setminus\{p_1,\ldots,p_m\}}$ is crossed a bounded number of times $C_0$ by the head of $T_1$ before time step ${kt_N}$.
  The claim that $\rho$ is crossed at most $C_0$ times before time step ${kt_N}$ follows since ${\{p_1,\ldots,p_m\}}$ are by definition empty before this time step.

  Finally, note that path ${p_1,\ldots,p_m}$ cannot move away more than distance $k$ from blocks ${b\otimes z_N}$ to ${b\otimes z_S}$, so if $n$ is large enough, the domain of $s_2$ is guaranteed to lie entirely inside the left connected component $A_0$ of ${\Z^2\setminus \rho}$.
  Similarly, the blocks containing the encoding of the rightmost copy of $u$ must lie entirely inside ${\Z^2\setminus A_0}$.
  In particular, the configuration $c$ of $T_1$ reached at step $kt_N$ must be such that ${\kolmo{c_{|\Z^2\setminus A_0}}\in\Omega(n)}$ by choice of $u$.
  However, Lemma~\ref{lem:kolmojordan} applied at step $kt_N$ to $T_1$ and $\rho$ gives: ${\kolmo{c_{|\Z^2\setminus A_0}}\leq C\log(kt_S)\in O(\log(n))}$ which is a contradiction for large enough $n$.
\end{proof}

\section{Universality Results}
\label{sec:univ}

\subsection{Radius 3 in 2D under Rigorous Simulations: the Heat Sink Trick}
\label{sec:2Dr3univ}

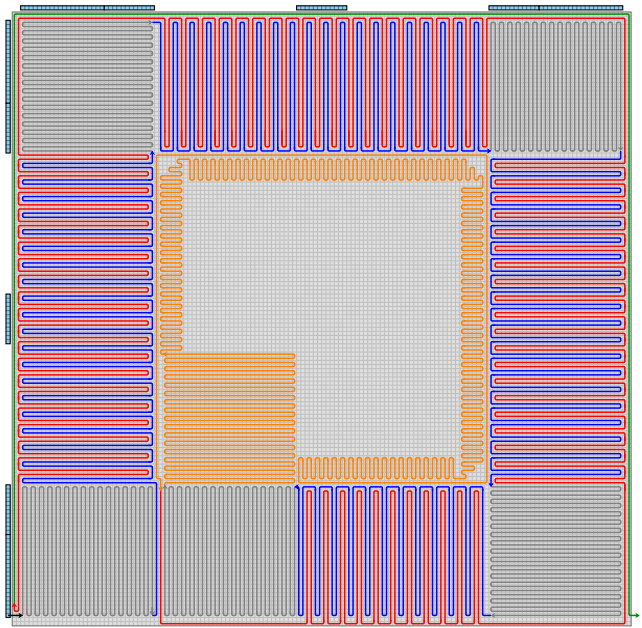
\begin{figure}[!ht]
\centering
\newcommand{\fullLineT}{.2mm}
\newcommand{\fullCorner}{.1ex}

\begin{tikzpicture}[scale = .055, 
arr/.style={->, -{Stealth[length=.5mm, width=.6mm]}},
backgroundLines/.style={lightgray, line width = .1mm},
cube/.style={fill=gray, fill opacity=1, draw opacity = 0, text = white},
cubeo/.style={fill=babyblue, fill opacity=1, draw opacity = 1, text = black},
cubet/.style={opacity = 0, text opacity = 1, text = black, font = \fontsize{4}{12} \selectfont},
write/.style={red, line width = \fullLineT, rounded corners=\fullCorner},
read/.style={blue, line width = \fullLineT, rounded corners=\fullCorner},
sortie/.style={islamicgreen, line width = \fullLineT, rounded corners=\fullCorner},
calcul/.style={orange, line width = \fullLineT, rounded corners=\fullCorner},
carr/.style={gray, line width = \fullLineT, rounded corners=\fullCorner, opacity = 1},
lettre/.style={fill=orange, fill opacity = .5, draw opacity = 0, text opacity = 1},
sepp/.style={black, opacity = .4}
]

\draw  [fill=lightgray ,fill opacity=.5, draw opacity = 0 ] (0,4) rectangle +(148,-148);
\foreach \x in {0,...,148} {
\draw [backgroundLines] (\x,4) -- (\x,-144);}
\foreach \x in {4,...,-144} {
\draw [backgroundLines] (0,\x) -- (148,\x);}
\draw  [gray ,fill opacity=0, draw opacity = 1 ] (0,4) -- (148,4) -- (148,-144) -- (0,-144) -- cycle;

\draw[cubeo] (2,4.5) rectangle +(20,1); \foreach \x in {3,...,21} {\draw [sepp] (\x,4.5) -- (\x,5.5);}
\draw[cubeo] (22,4.5) rectangle +(12,1); \foreach \x in {23,...,33} {\draw [sepp] (\x,4.5) -- (\x,5.5);}
\draw[cubeo] (68,4.5) rectangle +(12,1); \foreach \x in {69,...,79} {\draw [sepp] (\x,4.5) -- (\x,5.5);}
\draw[cubeo] (114,4.5) rectangle +(12,1); \foreach \x in {115,...,125} {\draw [sepp] (\x,4.5) -- (\x,5.5);}
\draw[cubeo] (126,4.5) rectangle +(20,1); \foreach \x in {127,...,145} {\draw [sepp] (\x,4.5) -- (\x,5.5);}

\draw[cubet] (2,6.5) rectangle +(20,1) node[pos= .5)] { transition table};
\draw[cubet] (22,6.5) rectangle +(12,1) node[pos=.5] { buffer};
\draw[cubet] (68,6.5) rectangle +(12,1) node[pos=.5] { $a_1 \in A_1$};
\draw[cubet] (114,6.5) rectangle +(12,1) node[pos=.5] { buffer};
\draw[cubet] (126,6.5) rectangle +(20,1) node[pos=.5] { transition table};

\draw[cubeo] (-.5,2) rectangle +(-1,-20); \foreach \y in {1,...,-17} {\draw [sepp] (-.5,\y) -- (-1.5,\y);}
\draw[cubeo] (-.5,-18) rectangle +(-1,-12); \foreach \y in {-19,...,-29} {\draw [sepp] (-.5,\y) -- (-1.5,\y);}
\draw[cubeo] (-.5,-64) rectangle +(-1,-12); \foreach \y in {-65,...,-75} {\draw [sepp] (-.5,\y) -- (-1.5,\y);}
\draw[cubeo] (-.5,-110) rectangle +(-1,-12); \foreach \y in {-111,...,-121} {\draw [sepp] (-.5,\y) -- (-1.5,\y);}
\draw[cubeo] (-.5,-122) rectangle +(-1,-20); \foreach \y in {-123,...,-141} {\draw [sepp] (-.5,\y) -- (-1.5,\y);}

\draw[cubet] (-3,-64) rectangle +(-1,-4) node[pos=.5] { $b_1$};
\draw[cubet] (-3,-68) rectangle +(-1,-4) node[pos=.5] { \rotatebox[origin=c]{-90}{$\in$}};
\draw[cubet] (-3,-72) rectangle +(-1,-4) node[pos=.5] { $A_1$};



\draw [carr] (1,-141.5) -- (2.5,-141.5) -- (2.5,-110.5) -- (3.5,-110.5) -- (3.5,-141); 
\foreach \x in {4,6,...,32}{
\draw [carr] (\x.5-1,-140.5) -- (\x.5-1,-141.5) -- (\x.5,-141.5) -- (\x.5,-110.5) -- (\x.5+1,-110.5) -- (\x.5+1,-141);}
\draw [carr, arr] (33.5,-139.5) -- (33.5,-141.5);

\foreach \x in {-29, -27, ..., -1}{
\draw [carr] (33.5,\x.5-.7) -- (33.5,\x.5) -- (2.5,\x.5) -- (2.5,\x.5+1) -- (33.5,\x.5+1) -- (33.5,\x.5+1.7);}
\draw [carr, arr] (33.5,0) -- (33.5,.5) -- (2.5,.5) -- (2.5,1.5) -- (33.5,1.5);

\foreach \x in {114,116,...,142}{
\draw [carr] (\x.5-.7,-29.5) -- (\x.5,-29.5) -- (\x.5,1.5) -- (\x.5+1,1.5) -- (\x.5+1,-29.5) -- (\x.5+1.7,-29.5);}
\draw [carr, arr] (144,-29.5) -- (144.5,-29.5) -- (144.5,1.5) -- (145.5,1.5) -- (145.5,-29.5);

\foreach \x in {-110,-112,...,-138}{
\draw [carr] (114.5,\x.5+.7) -- (114.5,\x.5) -- (145.5,\x.5) -- (145.5,\x.5-1) -- (114.5,\x.5-1) -- (114.5,\x.5-1.7);}
\draw [carr, arr] (114.5,-140) -- (114.5,-140.5) -- (145.5,-140.5) -- (145.5,-141.5) --(114.5,-141.5);

\foreach \x in {37,39,...,67}{
\draw [carr] (\x.5-1,-140.5) -- (\x.5-1,-141.5) -- (\x.5,-141.5) -- (\x.5,-110.5) -- (\x.5+1,-110.5) -- (\x.5+1,-141);}
\draw [carr, arr] (33.5,-139.5) -- (33.5,-141.5);
\draw [carr, arr] (36.5,-141.5) -- (36.5,-110);

\draw [sortie, black, arr] (-1,-141.5) -- (2.5, -141.5); 

\draw[sortie, arr] (.5,-139.5) -- (.5,3.5) -- (147.5,3.5) -- (147.5,-141.5) -- (150,-141.5);


\draw [read]  (33.5, -141.5) -- (34.5,-141.5) -- (34.5,-109.5) -- (2.5, -109.5) -- (2.5,-108.5) -- (33,-108.5);
\foreach \x in {-105, -101,...,-33} {
\draw [read] (32.5, \x.5-3) -- (33.5, \x.5-3) -- (33.5,\x.5) -- (2.5,\x.5) -- (2.5,\x.5+1) -- (33,\x.5+1);}
\draw [read, arr] (32.5,-32.5) -- (33.5,-32.5) -- (33.5,-29.5);

\draw [read]  (33.5,1.5) -- (35.5, 1.5) -- (35.5,-29);
\foreach \x in {38, 42,...,110} {
\draw [read] (\x.5-3,-28.5) -- (\x.5-3,-29.5) -- (\x.5,-29.5) -- (\x.5,1.5) -- (\x.5+1,1.5) -- (\x.5+1,-29);}
\draw [read, arr] (111.5,-28.5) -- (111.5,-29.5) -- (114.5,-29.5);

\draw [read]  (145.5, -29.5) -- (145.5,-31.5) -- (115,-31.5);
\foreach \x in {-107, -103,...,-35} {
\draw [read] (115,\x.5) -- (145.5,\x.5) -- (145.5,\x.5+1) -- (114.5,\x.5+1) -- (114.5,\x.5+4) -- (115.5,\x.5+4);}
\draw [read, arr] (115,-107.5) -- (114.5,-107.5) -- (114.5,-110.5);

\draw [read] (114.5,-141.5) -- (112.5,-141.5) -- (112.5,-111);
\foreach \x in {68, 72,...,110} {
\draw [read] (\x.5,-111) -- (\x.5,-141.5) -- (\x.5+1,-141.5) -- (\x.5+1,-110.5) -- (\x.5+4,-110.5) -- (\x.5+4,-111.5);}
\draw [read, arr] (68.5,-111.5) -- (68.5, -110.5) -- (67.5,-110.5);

\draw [write] (35.5,-110.5) -- (35.5,-143.5) -- (70.5, -143.5) -- (70.5, -142.5);

\foreach \x in {70,74,...,106}{
\draw [write] (\x.5,-143) -- (\x.5,-111.5) -- (\x.5+1, -111.5) -- (\x.5+1,-143.5) -- (\x.5+4,-143.5) -- (\x.5+4,-142.5);}
\draw [write] (110.5,-143) -- (110.5,-111.5) -- (111.5,-111.5) -- (111.5,-143.5) -- (146.5,-143.5) -- (146.5,-110);

\foreach \x in {-108,-104,...,-30}{
\draw [write] (146.5,\x.5+2) -- (146.5,\x.5) -- (115.5,\x.5) -- (115.5,\x.5-1) -- (146.5,\x.5-1) -- (146.5,\x.5-3);}
\draw [write] (112.5,-28) -- (112.5,-28.5) -- (113.5,-28.5) -- (113.5,2.5) -- (146.5,2.5) -- (146.5, -31);

\foreach \x in {37, 41,...,110} {
\draw [write] (\x.5-1, -24.5) -- (\x.5-1,-28.5) -- (\x.5,-28.5) -- (\x.5,2.5) -- (\x.5+3,2.5) -- (\x.5+3,-28);}
\draw [write] (1.5,-28.5) -- (1.5,2.5) -- (36.5,2.5) -- (36.5, -25);

\foreach \x in {-30,-34,...,-106}{
\draw [write] (1.5,\x.5+2) -- (1.5,\x.5) -- (32.5,\x.5) -- (32.5,\x.5-1) -- (1.5,\x.5-1) -- (1.5,\x.5-3);}
\draw [write, arr] (1.5,-108) -- (1.5,-140.5) -- (.5,-140.5) -- (.5,-138.5);



\foreach \x in {-109,-107,...,-81}{
\draw [calcul] (36.5,\x-1) -- (36.5,\x.5) -- (67.5,\x.5) -- (67.5,\x.5+1) -- (36.5,\x.5+1) -- (36.5,\x.5+1.7);}
\draw [calcul, arr] (36.5,-80) -- (36.5,-79.5) -- (67.5,-79.5) -- (67.5,-78.5) -- (36,-78.5);

\foreach \x in {-77,-75, ..., -37} {
\draw [calcul] (36.5,\x.5-1) -- (35.5, \x.5-1) -- (35.5, \x.5) -- (40.5,\x.5) -- (40.5,\x.5+1) -- (36,\x.5+1);} 

\draw [calcul] (36.5,-36.5) -- (35.5, -36.5) -- (35.5,-35.5) -- (40.5,-35.5) -- (40.5,-34.5) -- (37.5,-34.5) -- (37.5,-33.5) -- (40.5,-33.5) -- (40.5,-32.5) -- (39.5,-32.5) -- (39.5,-31.5) -- (40.5,-31.5) -- (42.5, -31.5) -- (42.5,-32.5);

\foreach \x in {43,45, ..., 107} {
\draw [calcul] (\x.5-1,-32) -- (\x.5-1,-36.5) -- (\x.5,-36.5) -- (\x.5, -31.5) -- (\x.5+1,-31.5) -- (\x.5+1, -32.5);} 

\draw [calcul] (108.5,-32) -- (108.5,-36.5) -- (109.5,-36.5) -- (109.5,-33.5) -- (110.5,-33.5) -- (110.5,-36.5) -- (111.5,-36.5) -- (111.5,-35.5) -- (112.5,-35.5) -- (112.5,-38.5) -- (108,-38.5);

\foreach \x in {-39,-41, ..., -103} {
\draw [calcul] (108.5, \x.5+1) -- (107.5, \x.5+1) -- (107.5,\x.5) -- (112.5,\x.5) -- (112.5, \x.5-1) -- (108,\x.5-1);}

\draw [calcul] (108.5, -104.5) -- (107.5,-104.5) -- (107.5,-105.5) -- (110.5,-105.5) -- (110.5,-106.5) -- (107.5,-106.5) -- (107.5,-107.5) -- (108.5,-107.5) -- (108.5,-108.5) -- (107.5,-108.5) -- (105.5,-108.5) -- (105.5,-104);

\foreach \x in {104,102, ..., 70} {
\draw [calcul] (\x.5+1,-104.5) -- (\x.5+1,-103.5) -- (\x.5,-103.5) -- (\x.5,-108.5) -- (\x.5-1,-108.5) -- (\x.5-1,-104);}

\draw [calcul, arr] (69.5,-104.5) -- (69.5,-103.5) -- (68.5,-103.5) -- (68.5,-109.5) -- (113.5,-109.5) -- (113.5,-30.5) -- (34.5,-30.5) -- (34.5,-108.5) -- (35.5,-108.5) -- (35.5,-111.5);

\end{tikzpicture}
\caption{\label{fig:TUR3univ_full} Behaviour, in one block, of the presented 2D radius 3 turedo $T_3 \in \univsim{2}{\simrig}{1}$.
The light blue rectangles represent the information to read of the neighbouring cells. In this example, a letter of $A_1$ is encoded by 3 letters of $A_3$. 
We can see the 2 padding cells, the transition table (here of arbitrary small size for readability of the figure), the buffer of size 12 (to contain 4 letters of $A_1$ encoded), the padding, the letter of the block encoded by 3 letters of $A_3$ repeated 4 times each and again the the padding, the buffer, the table and the last 2 padding cells.
The arrows represent the path followed by the turedo, entering the block in the bottom left of the figure.
The blue part reads the content of the adjacent blocks, the grey one allows for turning, the orange one is where the computing takes place, the red one fetchs the computed letters (and writes them on the faces that will not be visited again before exiting the block) and the green one finishes to write the letters and exits to the next block.
}
\end{figure}

Theorem~\ref{theo:noradius1univ} shows that no turedo of radius 1 can be universal for $\tur{2}{1}$ under rigorous simulations.
We show that this is however possible with radius 3.
In order to achieve universality, four key behaviours must be performed by our turedo at each simulation step. 
It needs to read all necessary information of neighbouring blocks, compute the next simulation step, write the computed letter and exit the current block (entering the correct next one).
Figure \ref{fig:TUR3univ_full} illustrate those behaviours. 
To rigorously simulate a turedo of radius 1, it seems natural to perform the 3 behaviours interacting with neighbours on the edge of the block, each on its layer, motivating a radius 3.
But as we are dealing with universality, the so called necessary information is not only the letters of neighbouring blocks but also the transition table defining the simulated turedo.
Carrying this information has a direct impact on the width of the reading and writing layer, we propose an intertwined zigzag to merge the space occupied by those two, keeping the radius 3 and taking full advantage of it.

\begin{theorem}\label{thm:heatsink}
  ${\univsim{2}{\simrig}{1}\cap\tur{2}{3}\neq\emptyset}$.
\end{theorem}
\begin{proof}
  We show that there is ${T_3\in\tur{2}{3}}$ such that for all ${T_1\in\tur{2}{1}}$: ${T_1\simrig T_3}$.
  Denote $T_3 = (A_3, Q_3, \delta_3)$. 
  Let's take $T_1 \in \tur{2}{1}$ some 2D turedo, $T_1 = (A_1, Q_1, \delta_1)$, a configuration $c_1 \in A_1^{\Z^2}$ and describe how $T_3$ simulates it with square blocks $\brect{(n,n)}$ and $n = 0 \bmod 4$.
	We first focus on the organisation of transmittable information in a given block $B$, \textit{i.e.} the transition table $\delta_1$ and the letter $a_1 \in A_1$ in this position in $c_1$.
	To be accessible to the neighbouring blocks, this information is present on the outside edges of $B$, repeated on each edge such that $B(0,i) = B(i,n-1) = B(n-1,n-1-i) = B(n-1-i,0)$.
	Considering a partial onto letter decoding map $\gamma : A_3^m \to A_1$ with $m \in \N$, the organisation on one edge of $B$ is the following. 
	The first two positions $B(0,0)$ and $B(0,1)$ are empty or irrelevant, then the next $3m|A_1|$ positions from $B(0,2)$ to $B(0,3m|A_1|+1)$ are the encoding of the transition table with $\gamma$ (which we assume to be a multiple of 4 without loss of generality). 
	Positions $B(0,3m|A_1|+2)$ to $B(0,m(3|A_1|+4)+1)$ are reserved for a buffer in which 4 letters will be encoded (the ones contained in the 4 neighbouring blocks). 
	Then $m(3|A_1|+4) +2$ positions are empty or irrelevant, from $B(0,m(3|A_1|+4)+2)$ to $B(0,2m(3|A_1|+4)+3)$, to allow the block to be spacious enough for the computation.
	Following that is written $u \in A_3^{4m}$ such that $u( 4i+k ) = \gamma(a_1)(i)$ for all $0 \leq k < 4$ (the redundancy is present to ensure proper reading later).
	Then again, from $B(0,2m(3|A_1|+6)+4)$ to $B(0,3m(3|A_1|+4)+2m+5)$ is some irrelevant padding followed by the $4m$ sized buffer, the $3m|A_1|$ sized transition table and 2 irrelevant position, finishing at $B(0,m(12|A_1|+20)+7)$. An illustration of this distribution is represented in figure \ref{fig:TUR3univ_full}.
					
	Let's now describe the behaviour of $T_3$ in a block to achieve universality. 
	As all necessary information to compute is contained in a $m(3|A_1|+4)$ letters long word on $A_3$, we base our construction on two types zigzags of this size : the square and the heat sink.
	A square gadget of even size $k$ is a back and forth $k/2$ times of a $k$ sized line during which the content of the original line is replicated on each four sides of the created square. 
	This copy is possible thanks to the radius of $T_3$ being greater than 2. Its main use is to keep and spread information while cornering.
	A heat sink gadget is a zigzag with each back and forth being spaced by 2. The radius 3 of $T_3$ allows the heat sink to still copy information from the previous zag during a zig.
	Its purpose is to acquire and transmit data laterally while being intertwined with another heat sink.
	Both gadget and the simulation of a block are illustrated in figure \ref{fig:TUR3univ_full} to illustrate the following explanation.
	
	$T_3$ enters a block $B$ at position $B(0,2)$ (or $B(2,n-1)$, $B(n,n-3)$, $B(n-3,0)$ up to rotation), it first continues forward by 2 (until reaching $B(2,2)$) then turns $90^{\circ}$ counter-clockwise and starts a square gadget, initialised by the copy of the transition table and empty buffer available at distance 3.
	Those squares will be performed at each corners of the block with in between a heat sink which will both transmit the transition table already collected earlier and approach the outside edge of the neighbouring block at distance 3, reading its information and filling the buffer. 
	(The heat sink has access to only half of this information but it has been taken care of by the redundancy detailed earlier).
	Once the information of all four neighbouring blocks collected, right after the last letter has been read, $T_3$ turns toward the center of the block, using the space left by the padding, performing one more square gadget.
	With a big enough transition table and a well chosen encoding, computing the next transition of $T_1$ can be performed in a square the size of the buffer and transition table encoded.
	As the only way to leave the center of the block is now a path of width 1, $T_3$ must transmit its computation before rejoining the edge of the block. 
	It does so with a zigzag that follows along the the inside of the reading heat sink.
	This information is then retrieved by the writing process, with a heat sink intertwined with the reading one. The writing process writes on the edge of the block on the faces before the exit one and at distance 1 after. 
	When following a square gadget, it copies the transition table and writes an empty buffer (which is possible because the square not only corners but copies information on all its sides), when following a reading gadget, it creates a heat sink gadget of its own, fetching the computed information and writing it on the face of the block.
	Lastly, once all the writing has been done, $T_3$ finishes filling its block by going back to the exit edge, following the writing path which had been shifted to the inside by 1 to allow for this, copying everything from the reading path to effectively write it on the edge making it attainable and finally $T_3$ exits the block at distance 2 from the corner.
\end{proof}

\subsection{The Power of  Third Dimension and Liberal Simulations}
\label{sec:3Duniv}

Having a third dimension available allows for a lot more freedom to simulate a turedo. Let us first focus on ${\tur{3}{1}}$ and rigorous simulations.
In dimension 2 we had to use the heat sink trick to gather all required information and we used a radius of 3 to accommodate for a reading, a writing and an exiting layer.
The third dimension allows us to shrink the radius to 1 thanks to its crossing capabilities.
The trick to achieve this is to have a marker to indicate if the neighbouring block is empty or not, to prevent trying to read in spaces where the turedo will have to write later on.

\begin{theorem}\label{thm:3Dunivr1rigor}
  ${\univsim{3}{\simrig}{1}\cap\tur{3}{1}\neq\emptyset}$.
\end{theorem}

\begin{proof}
	We show that there is ${T\in\tur{3}{1}}$, $T = (A,Q,\delta)$, such that for all ${T'\in\tur{3}{1}}$: ${T'\simrig T}$.
	Denote ${T' = (A',Q',\delta')}$.
	For a block $B$, the transmittable information is organised on the faces as follow. 
	In the center of each faces, a marker is written, indicating that information is present.
	On the face of the block facing the next block (the exit face), added to the presence marker, is an empty cross of width 1 with only at its end a stop marker indicating the end the block.
	Surrounding it is a cross of width 3 in which is written the transition table and a buffer big enough to accommodate for the encoding of 6 letters of $A'$ and one state from $Q'$ (with the size of the transition table and the buffer plus 1 left before writing). 
	On the face of the block facing other empty blocks, the same crossing pattern is used to only write the presence marker and the computed letter at distance 2 of the center of the face.
	See figure \ref{fig:TUR1univ3D} for this organisation on an example.

\newcommand{\cubeBack}[9]{
	\draw [#7,fill = #8] (#1+#3*\vectX,#2+#3*\vectY) -- (#1+#3*\vectX,#5+#3*\vectY) -- (#1+#6*\vectX,#5+#6*\vectY) -- (#4+#6*\vectX,#5+#6*\vectY) -- (#4+#6*\vectX,#2+#6*\vectY) -- (#4+#3*\vectX,#2+#3*\vectY) -- cycle; 
	\draw[#9] (#1+#3*\vectX,#2+#3*\vectY) -- (#1+#6*\vectX,#2+#6*\vectY) -- (#4+#6*\vectX,#2+#4*\vectY); 
	\draw[#9] (#1+#6*\vectX,#2+#6*\vectY) -- (#1+#6*\vectX,#5+#6*\vectY); 
}
\newcommand{\cubeFront}[7]{ 
	\draw [#7]  (#4+#6*\vectX,#5+#6*\vectY) -- (#4+#6*\vectX,#2+#6*\vectY) -- (#4+#3*\vectX,#2+#3*\vectY) -- (#1+#3*\vectX,#2+#3*\vectY) -- (#1+#3*\vectX,#5+#3*\vectY); 
	\draw[#7] (#1+#3*\vectX,#5+#3*\vectY) -- (#4+#3*\vectX,#5+#3*\vectY) -- (#4+#6*\vectX,#5+#6*\vectY); 
	\draw[#7] (#4+#3*\vectX,#2+#3*\vectY) -- (#4+#3*\vectX,#5+#3*\vectY) ; 
}
\newcommand{\cubeFull}[8]{
	\draw [#7,fill = #8] (#1+#3*\vectX,#2+#3*\vectY) -- (#1+#3*\vectX,#5+#3*\vectY) -- (#1+#6*\vectX,#5+#6*\vectY) -- (#4+#6*\vectX,#5+#6*\vectY) -- (#4+#6*\vectX,#2+#6*\vectY) -- (#4+#3*\vectX,#2+#3*\vectY) -- cycle; 
	\cubeFront{#1}{#2}{#3}{#4}{#5}{#6}{#7}
}

\newcommand{\milieuCompute}[1]{
	\draw[double arrow milieu=\arrpt pt colored by black and \colorCompute, rounded corners=\cornerex ex, ] (0+#1*\vectX+.5*\vectX,-32+#1*\vectY+.5*\vectY) -- (0+#1*\vectX,-32+#1*\vectY) -- (0+#1*\vectX,-16+#1*\vectY) -- (0+#1*\vectX-.5*\vectX,-16+#1*\vectY-.5*\vectY); 
	\draw[double arrow milieu=\arrpt pt colored by black and \colorCompute, rounded corners=\cornerex ex, ] (0+#1*\vectX-.5*\vectX,-16+#1*\vectY-.5*\vectY) -- (0+#1*\vectX-1*\vectX,-16+#1*\vectY-1*\vectY) -- (0+#1*\vectX-1*\vectX,-32+#1*\vectY-1*\vectY) -- (0+#1*\vectX-1.5*\vectX,-32+#1*\vectY-1.5*\vectY) ;}

\newcommand{\writeOne}[1]{\draw[double arrow milieu=\arrpt pt colored by black and \colorWrite, rounded corners=\cornerex ex, ] (0+18.5*\vectX,#1+18.5*\vectY) -- (0+19*\vectX,#1+19*\vectY) -- (0+19*\vectX,#1-1+19*\vectY) -- (0+18*\vectX,#1-1+18*\vectY) -- (0+18*\vectX,#1-2+18*\vectY) -- (0+18.5*\vectX,#1-2+18.5*\vectY);}

\newcommand{\writeTwo}[3]{ 
\draw[double arrow milieu=\arrpt pt colored by black and \colorWrite, rounded corners=\cornerex ex, ]  (#3+#1*\vectX+.5*\vectX,#2+#1*\vectY+.5*\vectY) -- (#3+#1*\vectX,#2+#1*\vectY) --  (#3-2+#1*\vectX,#2+#1*\vectY) --(#3-2+#1*\vectX-.5*\vectX,#2+#1*\vectY-.5*\vectY) ;
\draw[double arrow milieu=\arrpt pt colored by black and \colorWrite, rounded corners=\cornerex ex, ] (#3-2+#1*\vectX-.5*\vectX,#2+#1*\vectY-.5*\vectY) -- (#3-2+#1*\vectX-1*\vectX,#2+#1*\vectY-1*\vectY) -- (#3+#1*\vectX-1*\vectX,#2+#1*\vectY-1*\vectY) -- (#3+#1*\vectX-1.5*\vectX,#2+#1*\vectY-1.5*\vectY)  ;
}

\newcommand{\writeFour}[3]{ 
	\draw[double arrow milieu=\arrpt pt colored by black and \colorWrite, rounded corners=\cornerex ex, ] (#3+#1*\vectX+.5*\vectX, #2+#1*\vectY+.5*\vectY) -- (#3+#1*\vectX, #2+#1*\vectY) -- (#3+#1*\vectX, #2-2+#1*\vectY) -- (#3+#1*\vectX-.5*\vectX, #2-2+#1*\vectY-.5*\vectY);
	\draw[double arrow milieu=\arrpt pt colored by black and \colorWrite, rounded corners=\cornerex ex, ] (#3+#1*\vectX-.5*\vectX, #2-2+#1*\vectY-.5*\vectY) -- (#3+#1*\vectX-1*\vectX, #2-2+#1*\vectY-1*\vectY) -- (#3+#1*\vectX-1*\vectX, #2+#1*\vectY-1*\vectY) -- (#3+#1*\vectX-1.5*\vectX, #2+#1*\vectY-1.5*\vectY) ;}
	
	\newcommand{\writeThree}[1]{\draw[double arrow milieu=\arrpt pt colored by black and \colorWrite, rounded corners=\cornerex ex, ] (-1+1*\vectX,#1-.5+1*\vectY) -- (-1+1*\vectX,#1+1*\vectY) -- (-1+3*\vectX,#1+3*\vectY) -- (-1+3*\vectX,#1+1+3*\vectY) -- (-1+1*\vectX,#1+1+1*\vectY) -- (-1+1*\vectX,#1+1.5+1*\vectY);}

\newcommand{\finSque}[1]{\draw[double arrow milieu=\arrpt pt colored by black and \colorSque, rounded corners=\cornerex ex, ] (#1,-32.4-#1) -- (#1,-33-#1) -- (0,-33-#1) -- (0,-34-#1) -- (#1+2,-34-#1) -- (#1+2,-34.5-#1) ;}

\newcommand{\debutSque}[1]{\draw[double arrow milieu=\arrpt pt colored by black and \colorSque, rounded corners=\cornerex ex] (0,#1+.5) -- (0,#1) -- (0+32*\vectX, #1+32*\vectY) -- (0+32*\vectX,#1-1+32*\vectY) -- (0,#1-1) -- (0,#1-1.5) ;}

\newcommand{\toRightSque}[1]{\draw[double arrow milieu=\arrpt pt colored by black and \colorSque, rounded corners=\cornerex ex, ] (32-#1+#1*\vectX-\vectX,#1*\vectY-\vectY) -- (32-#1+#1*\vectX,#1*\vectY) -- (33-#1+#1*\vectX,#1*\vectY) -- (33-#1,0) -- (34-#1,0) -- (34-#1+#1*\vectX-3*\vectX,#1*\vectY-3*\vectY);}

\newcommand{\toLeftSque}[1]{\draw[double arrow milieu=\arrpt pt colored by black and \colorSque, rounded corners=\cornerex ex, ] (-34+\x-1.5*\vectX, 0-1.5*\vectY) -- (-34+\x-0*\vectX, 0-0*\vectY) -- (-33+\x-0*\vectX, 0-0*\vectY) -- (-33+\x-\x*\vectX, 0-\x*\vectY) -- (-32+\x-\x*\vectX, 0-\x*\vectY) -- (-32+\x-1*\vectX, 0-1*\vectY);}

\newcommand{\toFrontSque}[1]{
	\draw[double arrow milieu=\arrpt pt colored by black and \colorSque, rounded corners=\cornerex ex, ] (.4-#1*\vectX, 0-#1*\vectY) -- (32-#1-#1*\vectX, 0-#1*\vectY) -- (32-#1-#1*\vectX-\vectX, 0-#1*\vectY-\vectY) -- (30-#1-#1*\vectX-\vectX, 0-#1*\vectY-\vectY);
	\draw[double arrow milieu=\arrpt pt colored by black and \colorSque, rounded corners=\cornerex ex, ] (32-#1-#1*\vectX-.5*\vectX, 0-#1*\vectY-.5*\vectY) -- (32-#1-#1*\vectX-\vectX, 0-#1*\vectY-\vectY) -- (0-#1*\vectX-\vectX, 0-#1*\vectY-\vectY) -- (0-#1*\vectX-2*\vectX, 0-#1*\vectY-2*\vectY) -- (2-#1*\vectX-2*\vectX, 0-#1*\vectY-2*\vectY) ;}

\newcommand{\exitToFace}[3]{\draw[double arrow milieu=\arrpt pt colored by black and \colorExit, rounded corners=\cornerex ex, ] (#3-.5+#1*\vectX, #2-2+#1*\vectY) -- (#3+#1*\vectX, #2-2+#1*\vectY) -- (#3+#1*\vectX, #2+#1*\vectY) -- (#3+1+#1*\vectX, #2+#1*\vectY) -- (#3+1+#1*\vectX, #2-2+#1*\vectY) -- (#3+1.5+#1*\vectX, #2-2+#1*\vectY);}

\newcommand{\bottomTriangleExit}[1]{
	\draw[double arrow milieu=\arrpt pt colored by black and \colorExit, rounded corners=\cornerex ex, ] (32+#1*\vectX+16*\vectX+.5*\vectX, -32+#1+#1*\vectY+16*\vectY+.5*\vectY) -- (32+#1*\vectX+16*\vectX, -32+#1+#1*\vectY+16*\vectY) -- (32+#1*\vectX+16*\vectX, -16+#1*\vectY+16*\vectY) -- (32+#1*\vectX+16*\vectX-.5*\vectX, -16+#1*\vectY+16*\vectY-.5*\vectY) ;
	\draw[double arrow milieu=\arrpt pt colored by black and \colorExit, rounded corners=\cornerex ex, ] (32+#1*\vectX+16*\vectX-.5*\vectX, -16+#1*\vectY+16*\vectY-.5*\vectY) -- (32+#1*\vectX+16*\vectX-1*\vectX, -16+#1*\vectY+16*\vectY-1*\vectY) -- (32+#1*\vectX+16*\vectX-1*\vectX, -34+#1+#1*\vectY+16*\vectY-1*\vectY) -- (32+#1*\vectX+16*\vectX-1.5*\vectX, -34+#1+#1*\vectY+16*\vectY-1.5*\vectY) ;}
	
\newcommand{\bottomSquareExit}[1]{ 
	\draw[double arrow milieu=\arrpt pt colored by black and \colorExit, rounded corners=\cornerex ex, ] (32+#1*\vectX+.5*\vectX, -32+#1*\vectY+.5*\vectY) -- (32+#1*\vectX, -32+#1*\vectY) -- (32+#1*\vectX, -16+#1*\vectY) -- (32+#1*\vectX-.5*\vectX, -16+#1*\vectY-.5*\vectY) ;
	\draw[double arrow milieu=\arrpt pt colored by black and \colorExit, rounded corners=\cornerex ex, ] (32+#1*\vectX-.5*\vectX, -16+#1*\vectY-.5*\vectY) -- (32+#1*\vectX-1*\vectX, -16+#1*\vectY-1*\vectY) -- (32+#1*\vectX-1*\vectX, -32+#1*\vectY-1*\vectY) -- (32+#1*\vectX-1.5*\vectX, -32+#1*\vectY-1.5*\vectY) ;}
	
\newcommand{\backSquareExit}[1]{\draw[double arrow milieu=\arrpt pt colored by black and \colorExit, rounded corners=\cornerex ex, ] (32+32*\vectX, #1-.5+32*\vectY) -- (32+32*\vectX, #1+32*\vectY) -- (32+16*\vectX, #1+16*\vectY) -- (32+16*\vectX, #1+1+16*\vectY) -- (32+32*\vectX, #1+1+32*\vectY) -- (32+32*\vectX, #1+1.5+32*\vectY);}

\newcommand{\backTriangleExit}[1]{\draw[double arrow milieu=\arrpt pt colored by black and \colorExit, rounded corners=\cornerex ex, ] (32+32*\vectX-#1*\vectX, 17+#1-.5+32*\vectY-#1*\vectY) -- (32+32*\vectX-#1*\vectX, 17+#1+32*\vectY-#1*\vectY) -- (32+16*\vectX, 17+#1+16*\vectY) -- (32+16*\vectX, 17+#1+1+16*\vectY) -- (32+30*\vectX-#1*\vectX, 17+#1+1+30*\vectY-#1*\vectY) -- (32+30*\vectX-#1*\vectX, 17+#1+1.5+30*\vectY-#1*\vectY);}

\newcommand{\topSquareExit}[1]{ 
	\draw[double arrow milieu=\arrpt pt colored by black and \colorExit, rounded corners=\cornerex ex, ] (32+#1*\vectX+.5*\vectX, 32+#1*\vectY+.5*\vectY) -- (32+#1*\vectX, 32+#1*\vectY) -- (32+#1*\vectX, 16+#1*\vectY) -- (32+#1*\vectX-.5*\vectX, 16+#1*\vectY-.5*\vectY) ;
	\draw[double arrow milieu=\arrpt pt colored by black and \colorExit, rounded corners=\cornerex ex, ] (32+#1*\vectX-.5*\vectX, 16+#1*\vectY-.5*\vectY) -- (32+#1*\vectX-1*\vectX, 16+#1*\vectY-1*\vectY) -- (32+#1*\vectX-1*\vectX, 32+#1*\vectY-1*\vectY) -- (32+#1*\vectX-1.5*\vectX, 32+#1*\vectY-1.5*\vectY) ;}

\newcommand{\toFrontSquare}[1]{
	\draw[double arrow milieu=\arrpt pt colored by black and \colorExit, rounded corners=\cornerex ex, ] (16-\x*\vectX+.5*\vectX, 1-\x*\vectY+.5*\vectY) -- (16-\x*\vectX, 1-\x*\vectY) -- (32-\x*\vectX, 1-\x*\vectY) -- (32-\x*\vectX-.5*\vectX, 1-\x*\vectY-.5*\vectY);
	\draw[double arrow milieu=\arrpt pt colored by black and \colorExit, rounded corners=\cornerex ex, ] (32-\x*\vectX-.5*\vectX, 1-\x*\vectY-.5*\vectY) -- (32-\x*\vectX-1*\vectX, 1-\x*\vectY-1*\vectY) -- (16-\x*\vectX-1*\vectX, 1-\x*\vectY-1*\vectY) -- (16-\x*\vectX-1.5*\vectX, 1-\x*\vectY-1.5*\vectY);}
	
\newcommand{\toFrontTriangle}[1]{
	\draw[double arrow milieu=\arrpt pt colored by black and \colorExit, rounded corners=\cornerex ex, ]  (32-32*\vectX+#1*\vectX-.5*\vectX, 1-32*\vectY+#1*\vectY-.5*\vectY) -- (32-32*\vectX+#1*\vectX, 1-32*\vectY+#1*\vectY) -- (32 -#1-32*\vectX+#1*\vectX, 1-32*\vectY+#1*\vectY) -- (32 -#1-32*\vectX+#1*\vectX+.5*\vectX, 1-32*\vectY+#1*\vectY+.5*\vectY);
	\draw[double arrow milieu=\arrpt pt colored by black and \colorExit, rounded corners=\cornerex ex, ]  (32-32*\vectX+#1*\vectX-.5*\vectX, 1-32*\vectY+#1*\vectY-.5*\vectY) -- (32-32*\vectX+#1*\vectX-1*\vectX, 1-32*\vectY+#1*\vectY-1*\vectY) -- (34 -#1 -32*\vectX+#1*\vectX-1*\vectX, 1-32*\vectY+#1*\vectY-1*\vectY) -- (34 -#1 -32*\vectX+#1*\vectX-1.5*\vectX, 1-32*\vectY+#1*\vectY-1.5*\vectY) ;}
	
\newcommand{\frontSquareTriangleExit}[1]{ 
	\draw[double arrow milieu=\arrpt pt colored by black and \colorExit, rounded corners=\cornerex ex, ] (32-32*\vectX, -#1+.5-32*\vectY) -- (32-32*\vectX, -#1-32*\vectY) -- (32-16*\vectX, -#1-16*\vectY) -- (32-16*\vectX, -#1-1-16*\vectY) -- (32-32*\vectX, -#1-1-32*\vectY) -- (32-32*\vectX, -#1-1.5-32*\vectY) ;
	\draw[double arrow milieu=\arrpt pt colored by black and \colorExit, rounded corners=\cornerex ex, ] (32-32*\vectX+#1*\vectX, -16-#1+.5-32*\vectY+#1*\vectY) -- (32-32*\vectX+#1*\vectX, -16-#1-32*\vectY+#1*\vectY) -- (32-16*\vectX, -16-#1-16*\vectY) -- (32-16*\vectX, -16-#1-1-16*\vectY) -- (32-30*\vectX+#1*\vectX, -16-#1-1-30*\vectY+#1*\vectY) -- (32-30*\vectX+#1*\vectX, -16-#1-1.5-30*\vectY+#1*\vectY);}
	
\newcommand{\vectX}{.3}
\newcommand{\vectY}{.3}
\newcommand{\arrpt}{1}
\newcommand{\cornerex}{.05}
\newcommand{\colorSque}{azure(colorwheel)}
\newcommand{\colorCompute}{orange}
\newcommand{\colorExit}{green(ryb)}
\newcommand{\colorWrite}{tomato}

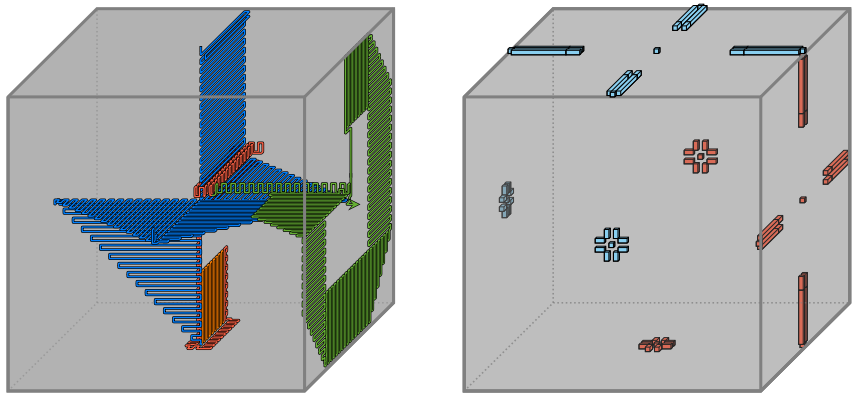
\begin{figure}
\centering
\begin{tikzpicture}[scale =.06,
arr/.style={->, >=stealth},
g/.style={gray, line width = .4mm, fill opacity = .3}, 
gprime/.style={gray, line width = .4mm, fill opacity = .6}, 
g2/.style={gray, line width = .2mm, densely dotted, opacity = .5}, 
b/.style={blue, line width = .1mm}, 
r/.style={red, line width = .1mm}, 
bl/.style={black, line width = .1mm}, 
] 

\cubeBack{-32.5}{-32.5}{-32.5}{32.5}{32.5}{32.5}{gprime}{gray}{g2}

\draw[double arrow milieu=\arrpt pt colored by black and \colorWrite, rounded corners=\cornerex ex, ]  (0+18.5*\vectX,-32+18.5*\vectY) --  (0+19*\vectX,-32+19*\vectY) --  (1+19*\vectX,-32+19*\vectY) -- (1+18*\vectX,-32+18*\vectY) -- (3+18*\vectX,-32+18*\vectY) -- (3+17.5*\vectX,-32+17.5*\vectY) ;

\foreach \y in {-30, -28, ..., -16}{\writeOne{\y}}
\draw[double arrow milieu=\arrpt pt colored by black and \colorWrite, rounded corners=\cornerex ex, ] (0+16.5*\vectX,-32+16.5*\vectY) -- (0+17*\vectX,-32+17*\vectY) -- (0+17*\vectX,-16+17*\vectY) -- (0+18.5*\vectX,-16+18.5*\vectY) ;

\foreach \z in {17,15,...,-1}{\writeTwo{\z}{-32}{3}}

\draw[double arrow milieu=\arrpt pt colored by black and \colorWrite, rounded corners=\cornerex ex, ]  (-1+1*\vectX,-30.5+1*\vectY) -- (-1+1*\vectX,-31+1*\vectY) -- (-1+3*\vectX,-31+3*\vectY) -- (-1+3*\vectX,-32+3*\vectY) -- (-1+2.5*\vectX,-32+2.5*\vectY);

\foreach \x in {-30,-28,...,-4}{\writeThree{\x}}

\draw[double arrow milieu=\arrpt pt colored by black and \colorWrite, rounded corners=\cornerex ex, ]  (-1+1*\vectX, -2.5+1*\vectY) --  (-1+1*\vectX, -2+1*\vectY) -- (-1+3*\vectX, -2+3*\vectY) -- (-1+3*\vectX, -1+3*\vectY) -- (-1+1*\vectX, -1+1*\vectY) -- (-0+1*\vectX, -1+1*\vectY) ;

\draw[double arrow milieu=\arrpt pt colored by black and \colorWrite, rounded corners=\cornerex ex, ] (1+31.5*\vectX, -1+31.5*\vectY) -- (1+32*\vectX, -1+32*\vectY) -- (3+32*\vectX, -1+32*\vectY) -- (3+32*\vectX, -2+32*\vectY) -- (1+32*\vectX, -2+32*\vectY) -- (1+32*\vectX, -3+32*\vectY) -- (1+32*\vectX, -3+32*\vectY) -- (0+32*\vectX, -3+32*\vectY) -- (0+32*\vectX, -1+32*\vectY) -- (-1+32*\vectX, -1+32*\vectY)-- (-1+32*\vectX, -3+32*\vectY) -- (-2+32*\vectX, -3+32*\vectY) -- (-3+32*\vectX, -3+32*\vectY) -- (-3+32*\vectX, -2+32*\vectY)  -- (-2+32*\vectX, -2+32*\vectY) -- (-2+32*\vectX, -1+32*\vectY) -- (-3+32*\vectX, -1+32*\vectY) -- (-3+32*\vectX, 0+32*\vectY) -- (-1+32*\vectX, 0+32*\vectY) -- (-1+32*\vectX, 1+32*\vectY) -- (-3+32*\vectX, 1+32*\vectY)-- (-3+32*\vectX, 2+32*\vectY) -- (-1+32*\vectX, 2+32*\vectY) -- (-1+32*\vectX, 3+32*\vectY) -- (-1+31.5*\vectX, 3+31.5*\vectY);
\draw[double arrow milieu=\arrpt pt colored by black and \colorWrite, rounded corners=\cornerex ex, ] (3+30.5*\vectX, -1+30.5*\vectY) -- (3+31*\vectX, -1+31*\vectY) -- (1+31*\vectX, -1+31*\vectY) -- (1+31.5*\vectX, -1+31.5*\vectY);
\foreach \z in {30,28,...,6}{\writeTwo{\z}{-1}{3}}
\draw[double arrow milieu=\arrpt pt colored by black and \colorWrite, rounded corners=\cornerex ex, ] (0+3.5*\vectX, -1+3.5*\vectY) -- (0+4*\vectX, -1+4*\vectY) -- (3+4*\vectX, -1+4*\vectY) -- (3+4.5*\vectX, -1+4.5*\vectY) ;
\draw[double arrow milieu=\arrpt pt colored by black and \colorWrite, rounded corners=\cornerex ex, ] (2+2.5*\vectX, -1+2.5*\vectY) -- (2+3*\vectX, -1+3*\vectY) -- (0+3*\vectX, -1+3*\vectY) -- (0+3.5*\vectX, -1+3.5*\vectY) ;
\draw[double arrow milieu=\arrpt pt colored by black and \colorWrite, rounded corners=\cornerex ex, ] (-0+1.5*\vectX, -1+1.5*\vectY) -- (-0+2*\vectX, -1+2*\vectY) -- (2+2*\vectX, -1+2*\vectY) -- (2+2.5*\vectX, -1+2.5*\vectY) ;
\draw[double arrow milieu=\arrpt pt colored by black and \colorWrite, rounded corners=\cornerex ex, ] (-1+1.5*\vectX, -1+1.5*\vectY) -- (-1+1*\vectX, -1+1*\vectY) -- (-0+1*\vectX, -1+1*\vectY) -- (-0+1.5*\vectX, -1+1.5*\vectY) ;

\foreach \z in {31, 29, ..., 1} {\writeFour{\z}{3}{-1}}

\draw[double arrow milieu=\arrpt pt colored by black and \colorWrite, rounded corners=\cornerex ex, ] (-3+1.5*\vectX,-32+1.5*\vectY) -- (-3+2*\vectX,-32+2*\vectY) -- (-1+2*\vectX,-32+2*\vectY) -- (-1+2.5*\vectX,-32+2.5*\vectY) ;
\draw[double arrow milieu=\arrpt pt colored by black and \colorWrite, rounded corners=\cornerex ex, ] (-3+1.5*\vectX,-32+1.5*\vectY) -- (-3+1*\vectX,-32+1*\vectY) -- (-1+1*\vectX,-32+1*\vectY) -- (-1+.5*\vectX,-32+.5*\vectY);
\draw[double arrow milieu=\arrpt pt colored by black and \colorWrite, rounded corners=\cornerex ex, ] (-3-.5*\vectX,-32-.5*\vectY) -- (-3-0*\vectX,-32-0*\vectY) -- (-1-0*\vectX,-32-0*\vectY) -- (-1+.5*\vectX,-32+.5*\vectY) ;
\draw[double arrow milieu=\arrpt pt colored by black and \colorWrite, rounded corners=\cornerex ex, ]  (-3-1*\vectX+.5*\vectX,-32-1*\vectY+.5*\vectY) -- (-3-1*\vectX,-32-1*\vectY) -- (-2-1*\vectX,-32-1*\vectY) --  (-2-3*\vectX,-32-3*\vectY) -- (-1-3*\vectX,-32-3*\vectY) --  (-1-1*\vectX,-32-1*\vectY) -- (0-1*\vectX,-32-1*\vectY) -- (0-3*\vectX,-32-3*\vectY) -- (3-3*\vectX,-32-3*\vectY) -- (3-3*\vectX+.5*\vectX,-32-3*\vectY+.5*\vectY);

\foreach \x in {16, 14, ..., 2}{\milieuCompute{\x}}


\foreach \x in {-30,-28,...,-3} { \finSque{\x}}
\draw[double arrow milieu=\arrpt pt colored by black and \colorSque, rounded corners=\cornerex ex, ] (-32+1*\vectX,-1+1*\vectY) -- (-32+3*\vectX,-1+3*\vectY) -- (-31+3*\vectX,-1+3*\vectY) -- (-31+0*\vectX,-1+0*\vectY) -- (-30+0*\vectX,-1+0*\vectY) -- (-30+1*\vectX,-1+1*\vectY) -- (-29+1*\vectX,-1+1*\vectY) -- (-29+0*\vectX,-1+0*\vectY) -- (0+0*\vectX,-1+0*\vectY)  -- (0+0*\vectX,-2+0*\vectY) -- (-30+0*\vectX,-2+0*\vectY) -- (-30+0*\vectX,-2.5+0*\vectY);
\draw[double arrow milieu=\arrpt pt colored by black and \colorSque, rounded corners=\cornerex ex, ] (-31-.5*\vectX,0-.5*\vectY) -- (-31-0*\vectX,0-0*\vectY) -- (-32-0*\vectX,0-0*\vectY)  -- (-32-0*\vectX,-1-0*\vectY) -- (-32+1*\vectX,-1+1*\vectY) ;
\draw[double arrow milieu=\arrpt pt colored by black and \colorSque, rounded corners=\cornerex ex] (-2,-30.5) -- (-2,-31) -- (0,-31) -- (0,-32) -- (0+.5*\vectX,-32+.5*\vectY);

\foreach \y in {30, 28, ..., 1} {\debutSque{\y}}
\draw[double arrow milieu=\arrpt pt colored by black and \colorSque, rounded corners=\cornerex ex] (0,.5) -- (0,0) -- (0+32*\vectX, 32*\vectY) -- (.5+32*\vectX,0+32*\vectY) ;
\draw[double arrow debut=\arrpt pt colored by black and \colorSque, rounded corners=\cornerex ex](0+28*\vectX,31+28*\vectY) -- (0+0*\vectX,31+0*\vectY) -- (0+0*\vectX,30.5+0*\vectY)  ; 
\draw[double arrow debut=\arrpt pt colored by black and \colorSque, rounded corners=\cornerex ex] (1+0*\vectX,31.5+0*\vectY) -- (1+0*\vectX,31+0*\vectY) -- (1+32*\vectX,31+32*\vectY)  -- (0+32*\vectX,31+32*\vectY)  -- (0+26*\vectX,31+26*\vectY);
\draw[double arrow debut=\arrpt pt colored by black and \colorSque, rounded corners=\cornerex ex] (0,32.5) -- (0,32)  -- (0+32*\vectX,32+32*\vectY) -- (1+32*\vectX,32+32*\vectY) -- (1+0*\vectX,32+0*\vectY) -- (1+0*\vectX,31.5+0*\vectY) ;
\draw[double arrow debut=\arrpt pt colored by black and \colorSque, rounded corners=\cornerex ex] (0,34) -- (0,32.5);

\foreach \x in {2,4,...,32}{  \toRightSque{\x}}
\draw[double arrow milieu=\arrpt pt colored by black and \colorSque, rounded corners=\cornerex ex, ] (30+1*\vectX,0+1*\vectY) -- (30+2*\vectX,0+2*\vectY) -- (31+2*\vectX,0+2*\vectY) -- (31,0) -- (32,0) -- (32-1*\vectX,0-1*\vectY) -- (29-1*\vectX,0-1*\vectY);
\draw[double arrow milieu=\arrpt pt colored by black and \colorSque, rounded corners=\cornerex ex, ] (32-.5*\vectX,0-.5*\vectY) -- (32-1*\vectX,0-1*\vectY) -- (0-1*\vectX,0-1*\vectY) -- (0-2*\vectX,0-2*\vectY) -- (1-2*\vectX,0-2*\vectY);

\foreach \x in {4,6,...,30}{  \toLeftSque{\x}}
\draw[double arrow milieu=\arrpt pt colored by black and \colorSque, rounded corners=\cornerex ex, ] (-30-1*\vectX, 0-1*\vectY) -- (-30-2*\vectX,0-2*\vectY) -- (-31-2*\vectX,0-2*\vectY) -- (-31-0*\vectX,0-0*\vectY) -- (-31.5-0*\vectX,0-0*\vectY);

\foreach \x in {2,4,...,29}{  \toFrontSque{\x}}
\draw[double arrow milieu=\arrpt pt colored by black and \colorSque, rounded corners=\cornerex ex, ] (.4-30*\vectX, 0-30*\vectY) -- (2-30*\vectX, 0-30*\vectY) -- (2-31*\vectX, 0-31*\vectY) -- (1-31*\vectX, 0-31*\vectY) ;

\draw[double arrow milieu=\arrpt pt colored by black and \colorSque, rounded corners=\cornerex ex, ](-1-28.5*\vectX, 0-28.5*\vectY) -- (-1-0*\vectX, 0-0*\vectY) -- (-2-0*\vectX, 0-0*\vectY) --(-2-1*\vectX, 0-1*\vectY);
\draw[double arrow milieu=\arrpt pt colored by black and \colorSque, rounded corners=\cornerex ex, ] (-1-29.5*\vectX, 1-29.5*\vectY) -- (-1-29*\vectX, 1-29*\vectY) -- (-1-29*\vectX, 0-29*\vectY) -- (-1-28.5*\vectX, 0-28.5*\vectY);
\draw[double arrow milieu=\arrpt pt colored by black and \colorSque, rounded corners=\cornerex ex, ] (-1-30.5*\vectX, 0-30.5*\vectY) -- (-1-30*\vectX, 0-30*\vectY) -- (-1-30*\vectX, 1-30*\vectY) -- (-1-29.5*\vectX, 1-29.5*\vectY) ;
\draw[double arrow milieu=\arrpt pt colored by black and \colorSque, rounded corners=\cornerex ex, ] (-1-31.5*\vectX, 3-31.5*\vectY)  -- (-1-31*\vectX, 3-31*\vectY) -- (-1-31*\vectX, 0-31*\vectY) -- (-1-30.5*\vectX, 0-30.5*\vectY) ;
\draw[double arrow milieu=\arrpt pt colored by black and \colorSque, rounded corners=\cornerex ex, ] (2-30.5*\vectX, 0-30.5*\vectY) -- (2-31*\vectX, 0-31*\vectY) -- (0-31*\vectX, 0-31*\vectY) -- (0-32*\vectX, 0-32*\vectY) -- (-1-32*\vectX, 0-32*\vectY) -- (-1-32*\vectX, 3-32*\vectY) -- (-1-31.5*\vectX, 3-31.5*\vectY) ;


\draw[double arrow milieu=\arrpt pt colored by black and \colorWrite, rounded corners=\cornerex ex, ] (1+31.5*\vectX, 1+31.5*\vectY) -- (1+32*\vectX, 1+32*\vectY) -- (1+32*\vectX, 3+32*\vectY) -- (2+32*\vectX, 3+32*\vectY) -- (2+32*\vectX, 1+32*\vectY) -- (3+32*\vectX, 1+32*\vectY) -- (3+32*\vectX, 3+32*\vectY) -- (4+32*\vectX, 3+32*\vectY) -- (4+32*\vectX, 1+32*\vectY) -- (4+31.5*\vectX, 1+31.5*\vectY);
\draw[double arrow milieu=\arrpt pt colored by black and \colorWrite, rounded corners=\cornerex ex, ] (1+30.5*\vectX, 3+30.5*\vectY) -- (1+31*\vectX, 3+31*\vectY) -- (1+31*\vectX, 1+31*\vectY) -- (1+31.5*\vectX, 1+31.5*\vectY) ;
\foreach \z in {30, 28, ..., 1} {\writeFour{\z}{3}{1}}
\draw[double arrow milieu=\arrpt pt colored by black and \colorWrite, rounded corners=\cornerex ex, ]  (1 -.5*\vectX, 1 -.5*\vectY) -- (1 +0*\vectX, 1 +0*\vectY) -- (1 +0*\vectX, 3 +0*\vectY) -- (1 +.5*\vectX, 3 +.5*\vectY);
\draw[double arrow milieu=\arrpt pt colored by black and \colorWrite, rounded corners=\cornerex ex, ] (-1 -.5*\vectX, 3 -.5*\vectY) -- (-1 -1*\vectX, 3 -1*\vectY) -- (-1 -1*\vectX, 1 -1*\vectY) -- (0 -1*\vectX, 1 -1*\vectY) -- (0 -1*\vectX, 3 -1*\vectY) -- (1 -1*\vectX, 3 -1*\vectY) -- (1 -1*\vectX, 1 -1*\vectY) -- (1 -.5*\vectX, 1 -.5*\vectY);

\draw[double arrow milieu=\arrpt pt colored by black and \colorWrite, rounded corners=\cornerex ex, ] (4+31.5*\vectX, 1+31.5*\vectY) -- (4+31*\vectX, 1+31*\vectY) -- (2+31*\vectX, 1+31*\vectY) -- (2+31*\vectX, 3+31*\vectY) -- (2+30.5*\vectX, 3+30.5*\vectY);

\foreach \z in {30,28, ..., 2} {\writeFour{\z}{3}{2}}

\draw[double arrow milieu=\arrpt pt colored by black and \colorWrite, rounded corners=\cornerex ex, ] (2+.5*\vectX, 3+.5*\vectY) -- (2-0*\vectX, 3-0*\vectY) -- (2-0*\vectX, 1-0*\vectY) -- (2.5-0*\vectX, 1-0*\vectY);


\draw[double arrow milieu=\arrpt pt colored by black and \colorExit, rounded corners=\cornerex ex, ] (32+3*\vectX,5+3*\vectY) -- (32+3*\vectX,0+3*\vectY) -- (32+2.5*\vectX,0+2.5*\vectY);
\draw[double arrow milieu=\arrpt pt colored by black and \colorExit, rounded corners=\cornerex ex, ] (32+2.5*\vectX,0+2.5*\vectY) -- (32+2*\vectX,0+2*\vectY) -- (32+2*\vectX,-1+2*\vectY) -- (32+1.5*\vectX,-1+1.5*\vectY);
\draw[double arrow fin=\arrpt pt colored by black and \colorExit, rounded corners=\cornerex ex, ] (32+1.5*\vectX,-1+1.5*\vectY) -- (32+0*\vectX,-1+0*\vectY) -- (35+0*\vectX,-1+0*\vectY);

\draw[double arrow milieu=\arrpt pt colored by black and \colorExit, rounded corners=\cornerex ex, ] (3+1.5*\vectX, 3+1.5*\vectY) -- (3+2*\vectX, 3+2*\vectY) -- (3+2*\vectX, 1+2*\vectY) -- (3.5+2*\vectX, 1+2*\vectY);
\draw[double arrow milieu=\arrpt pt colored by black and \colorExit, rounded corners=\cornerex ex, ] (2.5-0*\vectX, 1-0*\vectY) -- (3-0*\vectX, 1-0*\vectY) -- (3+1*\vectX, 1+1*\vectY) -- (3+1*\vectX, 3+1*\vectY) -- (3+1.5*\vectX, 3+1.5*\vectY);

\foreach \x in {4,6,...,28} {\exitToFace{2}{3}{\x}}
\draw[double arrow milieu=\arrpt pt colored by black and \colorExit, rounded corners=\cornerex ex, ] (29.5 + 2*\vectX, 1+2*\vectY) -- (30 + 2*\vectX, 1+2*\vectY) -- (30 + 2*\vectX, 3+2*\vectY) -- (31 + 2*\vectX, 3+2*\vectY) -- (31 + 2*\vectX, 1+2*\vectY) -- (32 + 2*\vectX, 1+2*\vectY) -- (32 + 2*\vectX, 3+2*\vectY) -- (32 + 1.5*\vectX, 3+1.5*\vectY);
\draw[double arrow milieu=\arrpt pt colored by black and \colorExit, rounded corners=\cornerex ex, ] (32 + 1.5*\vectX, 3+1.5*\vectY) -- (32 + 1*\vectX, 3+1*\vectY) -- (32 + 1*\vectX, 1+1*\vectY) -- (31 + 1*\vectX, 1+1*\vectY) -- (31 + 1*\vectX, 2+1*\vectY) -- (30 + 1*\vectX, 2+1*\vectY) -- (30 + 1*\vectX, 1+1*\vectY) -- (15 + 1*\vectX, 1+1*\vectY) -- (15 + .5*\vectX, 1+.5*\vectY);

\draw[double arrow milieu=\arrpt pt colored by black and \colorExit, rounded corners=\cornerex ex, ] (32+30.5*\vectX, -18+30.5*\vectY)  -- (32+31*\vectX, -18+31*\vectY) -- (32+31*\vectX, -16+31*\vectY) -- (32+32*\vectX, -16+32*\vectY) -- (32+32*\vectX, -15+32*\vectY) -- (32+31*\vectX, -15+31*\vectY);
\foreach \x in {14, 12, ...,2}{ \bottomTriangleExit{\x}}

\foreach \x in {16, 14,...,-14}{ \bottomSquareExit{\x}}


\foreach \x in {-15, -13,...,15}{ \backSquareExit{\x}}
\foreach \x in {0, 2,...,12}{ \backTriangleExit{\x}}
\draw[double arrow milieu=\arrpt pt colored by black and \colorExit, rounded corners=\cornerex ex, ] (32+18*\vectX,30.5+18*\vectY) --  (32+18*\vectX,31+18*\vectY) -- (32+16*\vectX,31+16*\vectY) -- (32+16*\vectX,32+16*\vectY) -- (32+15.5*\vectX,32+15.5*\vectY) ;

\foreach \x in {15, 13,...,0}{ \topSquareExit{\x}}

\draw[double arrow milieu=\arrpt pt colored by black and \colorExit, rounded corners=\cornerex ex, ] (32-.5*\vectX,32-.5*\vectY) -- (32-1*\vectX,32-1*\vectY) -- (32-1*\vectX,15-1*\vectY) -- (32+3*\vectX,15+3*\vectY) -- (32+3*\vectX,4+3*\vectY);

\draw[double arrow milieu=\arrpt pt colored by black and \colorExit, rounded corners=\cornerex ex, ] (15+.5*\vectX, 1+.5*\vectY) -- (15,1) -- (32,1) -- (32-.5*\vectX, 1-.5*\vectY);
\draw[double arrow milieu=\arrpt pt colored by black and \colorExit, rounded corners=\cornerex ex, ] (32-.5*\vectX, 1-.5*\vectY) -- (32-1*\vectX, 1-1*\vectY) -- (16-1*\vectX, 1-1*\vectY) -- (16-1.5*\vectX, 1-1.5*\vectY);
\foreach \x in {2,4,...,14}{\toFrontSquare{\x}}

\foreach \x in {16,14,...,4}{  \toFrontTriangle{\x}}
\draw[double arrow milieu=\arrpt pt colored by black and \colorExit, rounded corners=\cornerex ex, ] (30-29.5*\vectX,1-29.5*\vectY) --(30-30*\vectX,1-30*\vectY) -- (32-30*\vectX,1-30*\vectY) -- (32-30.5*\vectX,1-30.5*\vectY) ;
\draw[double arrow milieu=\arrpt pt colored by black and \colorExit, rounded corners=\cornerex ex, ] (32-30.5*\vectX,1-30.5*\vectY) -- (32-31*\vectX,1-31*\vectY) -- (32-32*\vectX,1-32*\vectY) -- (32-32*\vectX,.5-32*\vectY);

\foreach \x in {0, 2,...,14}{ \frontSquareTriangleExit{\x}}
\draw[double arrow milieu=\arrpt pt colored by black and \colorExit, rounded corners=\cornerex ex, ] (32-16*\vectX, -31-16*\vectY)  -- (32-16*\vectX, -32-16*\vectY) -- (32-15.5*\vectX, -32-15.5*\vectY);

\cubeFront{-32.5}{-32.5}{-32.5}{32.5}{32.5}{32.5}{g}


\begin{scope}[shift={(100,0)}]


\cubeFull{-33.5}{.5}{1.5}{-32.5}{1.5}{3.5}{bl}{babyblue}
\cubeFull{-33.5}{-1.5}{1.5}{-32.5}{-.5}{3.5}{bl}{babyblue}
\cubeFull{-33.5}{-3.5}{.5}{-32.5}{-1.5}{1.5}{bl}{babyblue}
\cubeFull{-33.5}{-3.5}{-1.5}{-32.5}{-1.5}{-.5}{bl}{babyblue}
\cubeFull{-33.5}{-.5}{-.5}{-32.5}{.5}{.5}{bl}{babyblue}
\cubeFull{-33.5}{1.5}{.5}{-32.5}{3.5}{1.5}{bl}{babyblue}
\cubeFull{-33.5}{1.5}{-1.5}{-32.5}{3.5}{-.5}{bl}{babyblue}
\cubeFull{-33.5}{.5}{-3.5}{-32.5}{1.5}{-1.5}{bl}{babyblue}
\cubeFull{-33.5}{-1.5}{-3.5}{-32.5}{-.5}{-1.5}{bl}{babyblue}

\cubeBack{-32.5}{-32.5}{-32.5}{32.5}{32.5}{32.5}{g}{gray}{g2}


\cubeFull{-3.5}{.5}{31.5}{-1.5}{1.5}{32.5}{bl}{tomato}
\cubeFull{-3.5}{-1.5}{31.5}{-1.5}{-.5}{32.5}{bl}{tomato}
\cubeFull{-1.5}{-3.5}{31.5}{-.5}{-1.5}{32.5}{bl}{tomato}
\cubeFull{.5}{-3.5}{31.5}{1.5}{-1.5}{32.5}{bl}{tomato}
\cubeFull{-.5}{-.5}{31.5}{.5}{.5}{32.5}{bl}{tomato}
\cubeFull{-1.5}{1.5}{31.5}{-.5}{3.5}{32.5}{bl}{tomato}
\cubeFull{.5}{1.5}{31.5}{1.5}{3.5}{32.5}{bl}{tomato}
\cubeFull{1.5}{.5}{31.5}{3.5}{1.5}{32.5}{bl}{tomato}
\cubeFull{1.5}{-1.5}{31.5}{3.5}{-.5}{32.5}{bl}{tomato}

\cubeFull{-1.5}{-32.5}{1.5}{-.5}{-31.5}{3.5}{bl}{tomato}
\cubeFull{.5}{-32.5}{1.5}{1.5}{-31.5}{3.5}{bl}{tomato}
\cubeFull{-3.5}{-32.5}{.5}{-1.5}{-31.5}{1.5}{bl}{tomato}
\cubeFull{-3.5}{-32.5}{-1.5}{-1.5}{-31.5}{-.5}{bl}{tomato}
\cubeFull{-.5}{-32.5}{-.5}{.5}{-31.5}{.5}{bl}{tomato}
\cubeFull{1.5}{-32.5}{.5}{3.5}{-31.5}{1.5}{bl}{tomato}
\cubeFull{1.5}{-32.5}{-1.5}{3.5}{-31.5}{-.5}{bl}{tomato}
\cubeFull{-1.5}{-32.5}{-3.5}{-.5}{-31.5}{-1.5}{bl}{tomato}
\cubeFull{.5}{-32.5}{-3.5}{1.5}{-31.5}{-1.5}{bl}{tomato}

\cubeFull{31.5}{-.5}{31.5}{32.5}{.5}{32.5}{bl}{tomato}
\cubeFull{31.5}{-32.5}{-.5}{32.5}{-31.5}{.5}{bl}{tomato}

\cubeFull{31.5}{.5}{19.5}{32.5}{1.5}{31.5}{bl}{tomato}
\cubeFull{31.5}{.5}{16.5}{32.5}{1.5}{19.5}{bl}{tomato}
\cubeFull{31.5}{-1.5}{19.5}{32.5}{-.5}{31.5}{bl}{tomato}
\cubeFull{31.5}{-1.5}{16.5}{32.5}{-.5}{19.5}{bl}{tomato}
\cubeFull{31.5}{-31.5}{.5}{32.5}{-19.5}{1.5}{bl}{tomato}
\cubeFull{31.5}{-19.5}{.5}{32.5}{-16.5}{1.5}{bl}{tomato}
\cubeFull{31.5}{-31.5}{-1.5}{32.5}{-19.5}{-.5}{bl}{tomato}
\cubeFull{31.5}{-19.5}{-1.5}{32.5}{-16.5}{-.5}{bl}{tomato}
\cubeFull{31.5}{-.5}{-.5}{32.5}{.5}{.5}{bl}{tomato}
\cubeFull{31.5}{16.5}{.5}{32.5}{19.5}{1.5}{bl}{tomato}
\cubeFull{31.5}{19.5}{.5}{32.5}{31.5}{1.5}{bl}{tomato}
\cubeFull{31.5}{16.5}{-1.5}{32.5}{19.5}{-.5}{bl}{tomato}
\cubeFull{31.5}{19.5}{-1.5}{32.5}{31.5}{-.5}{bl}{tomato}
\cubeFull{31.5}{31.5}{-.5}{32.5}{32.5}{.5}{bl}{tomato} 
\cubeFull{31.5}{.5}{-19.5}{32.5}{1.5}{-16.5}{bl}{tomato}
\cubeFull{31.5}{.5}{-31.5}{32.5}{1.5}{-19.5}{bl}{tomato}
\cubeFull{31.5}{-1.5}{-19.5}{32.5}{-.5}{-16.5}{bl}{tomato}
\cubeFull{31.5}{-1.5}{-31.5}{32.5}{-.5}{-19.5}{bl}{tomato}
\cubeFull{31.5}{-.5}{-32.5}{32.5}{.5}{-31.5}{bl}{tomato} 

\cubeBack{-32.5}{-32.5}{-32.5}{32.5}{32.5}{32.5}{g}{gray}{g2}
\cubeFront{-32.5}{-32.5}{-32.5}{32.5}{32.5}{32.5}{g}

\cubeFull{-1.5}{32.5}{19.5}{-.5}{33.5}{31.5}{bl}{babyblue}
\cubeFull{-1.5}{32.5}{16.5}{-.5}{33.5}{19.5}{bl}{babyblue}
\cubeFull{-.5}{32.5}{31.5}{.5}{33.5}{32.5}{bl}{babyblue} 
\cubeFull{.5}{32.5}{19.5}{1.5}{33.5}{31.5}{bl}{babyblue}
\cubeFull{.5}{32.5}{16.5}{1.5}{33.5}{19.5}{bl}{babyblue}
\cubeFull{-31.5}{32.5}{.5}{-19.5}{33.5}{1.5}{bl}{babyblue}
\cubeFull{-19.5}{32.5}{.5}{-16.5}{33.5}{1.5}{bl}{babyblue}
\cubeFull{-32.5}{32.5}{-.5}{-31.5}{33.5}{.5}{bl}{babyblue} 
\cubeFull{-31.5}{32.5}{-1.5}{-19.5}{33.5}{-.5}{bl}{babyblue}
\cubeFull{-19.5}{32.5}{-1.5}{-16.5}{33.5}{-.5}{bl}{babyblue}
\cubeFull{-.5}{32.5}{-.5}{.5}{33.5}{.5}{bl}{babyblue}
\cubeFull{16.5}{32.5}{.5}{19.5}{33.5}{1.5}{bl}{babyblue}
\cubeFull{19.5}{32.5}{.5}{31.5}{33.5}{1.5}{bl}{babyblue}
\cubeFull{31.5}{32.5}{-.5}{32.5}{33.5}{.5}{bl}{babyblue} 
\cubeFull{16.5}{32.5}{-1.5}{19.5}{33.5}{-.5}{bl}{babyblue}
\cubeFull{19.5}{32.5}{-1.5}{31.5}{33.5}{-.5}{bl}{babyblue}
\cubeFull{-1.5}{32.5}{-19.5}{-.5}{33.5}{-16.5}{bl}{babyblue}
\cubeFull{-1.5}{32.5}{-31.5}{-.5}{33.5}{-19.5}{bl}{babyblue}
\cubeFull{-.5}{32.5}{-32.5}{.5}{33.5}{-31.5}{bl}{babyblue} 
\cubeFull{.5}{32.5}{-19.5}{1.5}{33.5}{-16.5}{bl}{babyblue}
\cubeFull{.5}{32.5}{-31.5}{1.5}{33.5}{-19.5}{bl}{babyblue}

\cubeFull{-3.5}{.5}{-33.5}{-1.5}{1.5}{-32.5}{bl}{babyblue}
\cubeFull{-3.5}{-1.5}{-33.5}{-1.5}{-.5}{-32.5}{bl}{babyblue}
\cubeFull{-1.5}{-3.5}{-33.5}{-.5}{-1.5}{-32.5}{bl}{babyblue}
\cubeFull{.5}{-3.5}{-33.5}{1.5}{-1.5}{-32.5}{bl}{babyblue}
\cubeFull{-.5}{-.5}{-33.5}{.5}{.5}{-32.5}{bl}{babyblue}
\cubeFull{-1.5}{1.5}{-33.5}{-.5}{3.5}{-32.5}{bl}{babyblue}
\cubeFull{.5}{1.5}{-33.5}{1.5}{3.5}{-32.5}{bl}{babyblue}
\cubeFull{1.5}{.5}{-33.5}{3.5}{1.5}{-32.5}{bl}{babyblue}
\cubeFull{1.5}{-1.5}{-33.5}{3.5}{-.5}{-32.5}{bl}{babyblue}

\end{scope}

\end{tikzpicture}
\caption{\label{fig:TUR1univ3D}
One block of the presented 3D radius 1 turedo $T \in \univsim{3}{\simrig}{1}$.
At this stage $T$ comes from the block above, the front and left blocks are also non-empty and the others are empty. $T$ will exit the block by entering the next block at its right.
On the right is represented the disposition of information. The light blue cubes are information to read in adjacent blocks and the red ones are information written by $T$ in the considered block.
In this example, a letter of $A'$ is encoded by 2 letters of $A$, the transition table is of arbitrary small size for readability and they are written in a cross on the faces of blocks.
On the left is the behaviour of $T$ in the same considered block, in blue is the reading phase (building also a skeleton), in orange the computation phase, in red the writing phase and in green the exit phase.}
\end{figure}

	  $T$ has the following behaviour in a block, also illustrated in figure \ref{fig:TUR1univ3D}. First, $T$ performs a reading phase : entering a block by the middle of a face, it travels straight to the edge of the face until reading the stop marker. We assume $T$ entering by the middle of the face as it truly enters at distance one of the middle but this shift of one position in  one direction can be remembered and corrected immediately after entering.
	  This allows $T$ to turn and go back toward the center, reading the transition table and the buffer. 
	  All this information gathered, $T$ follows a fixed path composed of shrinking zigzags, staying in the planes formed by the 3D cross centered in the block (hence the peculiar way we place information, to reserve space for this path).
	  Doing so, it reaches the center of each face while carrying the transition table and filling the transition buffer when a presence marker is read.
	  As the order in which the faces are visited is fixed, $T$ can store in its head state which neighbouring blocks are empty.
	  This reading phase has the added benefit to have created a skeleton in the block, enabling easy travel for the following phases.
	  Next $T$ enters the computing phase consisting of a square zigzag, using the transition table and the now filled buffer carried previously and so present in all the branches of the centered 3D cross.
	  Once the next step of the simulation computed, $T$ can carry the letter to encode following the skeleton and write it on all previously identified faces except the exit one.
	  The last phase, the exit one, consists in carrying the encoded letter to the exit face, retrieve on the centered 3D cross the transition table and the buffer (reset with the new state and letter of the computed transition) and writing all this information as previously presented thanks to square and triangular zigzags. 
	  Finally $T$ exits to the next block, at distance one of the center of the face (to get around the reading skeleton).
\end{proof}

The combination of 3D and liberal simulations allows to shrink the radius of any turedo to $1$.
In this construction, the computation of simulated transitions is done internally in the simulating turedo's head.
The challenging part however is in acquiring the states of distant neighbors.
Thankfully the liberal nature of the simulation allows travel through empty blocks and the 3D enables crossing paths without intersection.
Still, a rigorous organisation is needed in order to prevent overlapping.

\begin{theorem}\label{thm:3Dradiusliberal}
  For any radius $r$ and any ${T_r\in\tur{3}{r}}$ there is ${T_1\in\tur{3}{1}}$ such that ${T_r\simlib T_1}$.
\end{theorem}

\begin{proof}
Let $r \in \N_+$ and $T_r \in \tur{3}{r}$, $T_r = (A_r, Q_r, \delta_r)$. 
We build $T_1 = (A_1, Q_1, \delta_1)$, with $Q_1$ big enough to encode in one state a position $z \in \Z^3$ modulo $2r$ of each dimension and the $|\boule{3}{r}|$ neighbouring letters.
This allows for an instant computation of $\delta_r$ once all needed information is gathered.
Let $b$ be the block size, the critical aspect of this simulation is for $T_1$ to visit all blocks ${b\otimes z'}$ with $z' \in z+\boule{3}{r}$ for each simulation step.
Therefore we have to assign non intersecting exploration paths for all positions at distance less than $2r$. 
To achieve this, we define $\mathcal{C} = \lbrace{0, ..., 8r^3 -1}\rbrace$ a set of colors and we assign the color ${(z_1 \bmod 2r)} + {2r (z_2 \bmod 2r)} + {4r^2 (z_3 \bmod 2r)}$ to the block ${b\otimes (z_1,z_2,z_3)}$.
By taking $b = 8r^3l+7$ with $l \in \N_+$, in each block, for each color $c \in \mathcal{C}$, we can reserve tubes of width $l$ following the edges of a centered cube of edges of length $lc+1$ and the direct extension of said edges to the face of the block (see figure \ref{cube_nested_cubes}). 
This creates reserved spaces for each color consisting of centered nested cubes. 
Those cubes fill a space of $8r^3l$, we add $1$ to have a proper center and $6$ to have some padding near the faces of the blocks (we will discuss its necessity later). 
Note that $l$ actually doesn't need to be large, $l=10$ is enough.
On each face of the block ${b\otimes z}$, the letter $a \in A_r \subset A_1$ is repeated in a cross pattern (see \ref{cube_cross}).

\begin{figure}[h]
    \begin{subfigure}[t]{0.35\textwidth}
    \center

\begin{tikzpicture}[scale = .2,
b/.style={blue, line width=1mm}, 
b2/.style={blue, line width=.6mm, densely dotted, opacity = .5}, 
g/.style={gray, line width = .5mm}, 
g2/.style={gray, line width = .3mm, densely dotted, opacity = .5}, 
t/.style={line width=.5mm, densely dotted, ->, >=stealth,opacity = 0}] 

\draw [g,fill = lightgray]  (0,0) -- (0,10) -- (4,13) -- (14,13) -- (14,3) -- (10,0) -- cycle ; 

\draw[g2] (0,0) -- (4,3) -- (14,3); 
\draw[g2] (4,3) -- (4,13); 

\draw[b2] (5,0) -- (9,3) -- (9,13); 
\draw[b2] (0,5) -- (4,8) -- (14,8); 
\draw[b2] (2,11.5) -- (2,1.5) -- (12,1.5); 
\draw[b] (0,5) -- (10,5) -- (14,8); 
\draw[b] (5,0) -- (5,10) -- (9,13); 
\draw[b] (2,11.5) -- (12,11.5) -- (12,1.5); 

\draw[g]  (0,0) -- (0,10) -- (4,13) -- (14,13) -- (14,3) -- (10,0) -- cycle ; 
\draw[g]   (0,10) -- (10,10) -- (14,13) ; 
\draw[g]  (10,0) -- (10,10) ; 

\end{tikzpicture}
\caption{\label{cube_cross} The letter $a \in A_r \subset A_1$ is repeated on all faces in a cross pattern.}
    \end{subfigure}
    \hfill
    \begin{subfigure}[t]{0.34\textwidth}
    \center

\begin{tikzpicture}[scale = .2,
r/.style={red, line width=2mm}, 
r2/.style={red, line width=1.5mm, densely dotted, opacity = .5}, 
v/.style={islamicgreen, line width=2mm}, 
v2/.style={islamicgreen, line width=.5mm, densely dotted, ->, >=stealth}, 
g/.style={gray, line width = .5mm}, 
g2/.style={gray, line width = .3mm, densely dotted, opacity = .5}] 

\draw [g,fill = lightgray]  (0,0) -- (0,10) -- (4,13) -- (14,13) -- (14,3) -- (10,0) -- cycle ; 

\draw[g2] (0,0) -- (4,3) -- (14,3); 
\draw[g2] (4,3) -- (4,13); 

\draw[v]   (2.5,7.5) -- (2.5+1,7.5+.75) -- (2.5+1,10+.75) ; 
\draw[v]  (2.5+1,7.5+.75) -- (0+1,7.5+.75) ; 
\draw[v2]  (2.5+1,7.5+.75) -- (-4+1,7.5+.75) ; 
\draw[v2]  (2.5+1,7.5+.75) -- (2.5-1.6,7.5-1.2) ; 

\draw[r]  (2.5+1,2.5+.75) -- (2.5+1,7.5+.75) -- (2.5+3,7.5+2.25) -- (7.5+3,7.5+2.25) -- (7.5+3,2.5+2.25) -- (7.5+1,2.5+.75) -- cycle ;
\draw[r]  (2.5+1,7.5+.75) -- (7.5+1,7.5+.75) -- (7.5+3,7.5+2.25) ; 
\draw[r]  (7.5+1,2.5+.75) -- (7.5+1,7.5+.75) ; 
\draw[r2] (2.5+1,2.5+.75) -- (2.5+3,2.5+2.25) -- (7.5+3,2.5+2.25); 
\draw[r2] (2.5+3,2.5+2.25) -- (2.5+3,2.5+7.25); 

\draw[v]   (2.5,7.5) -- (2.5+1,7.5+.75);

\draw[g]  (0,0) -- (0,10) -- (4,13) -- (14,13) -- (14,3) -- (10,0) -- cycle ; 
\draw[g]   (0,10) -- (10,10) -- (14,13) ; 
\draw[g]  (10,0) -- (10,10) ; 

\draw[v2] (2.5+1,9.5+.75) -- (2.5+1, 14+.75); 

\end{tikzpicture}
\caption{\label{cube_nested_cubes} Reserved space for the exploration path of color $c \in \mathcal{C}$.}

    \end{subfigure}
	\hfill
	\begin{subfigure}[t]{0.25\textwidth}
	\center
	
\begin{tikzpicture}[scale = .2,
r/.style={red, line width=2mm}, 
b/.style={blue, line width=1mm}, 
v/.style={islamicgreen, line width=2mm}, 
g/.style={gray, line width = .5mm}, 
t/.style={line width=.5mm, densely dotted, ->, >=stealth,opacity = 0}] 

\draw [g,fill = lightgray]  (0,0) -- (0,10) -- (10,10) -- (10,0) -- cycle ; 

\draw[v] (2.5,2.5) -- (7.5,2.5) -- (7.5,7.5) -- (2.5,7.5) -- cycle;

\draw[b] (0,5) -- (10,5);
\draw[b] (5,0) -- (5,10);

\end{tikzpicture}
\caption{\label{cube_face} One face of a block.}

	\end{subfigure}   
\caption{\label{cube} Representation of the reserved space for color $c$ in block ${b\otimes z}$. 
In red are the tubes of width $l$ on the edges of the centered cube of edges of length $lc+1$.
In green are the extensions of said tubes, allowing to reach the reserved space for color $c$ in the neighbouring blocks, completing the exploration path of color $c$ (\ref{cube_nested_cubes}). Those extensions also allow the exploration path to access the blue cross containing the letter of position $z$ in the configuration (to either read or write)(\ref{cube_face}).}

  \end{figure}
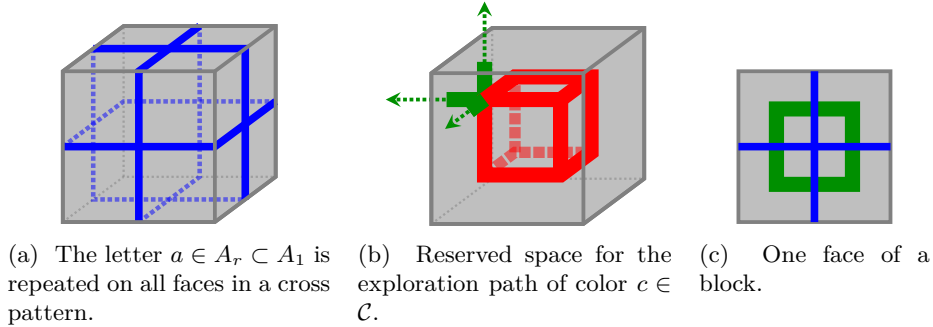

$F_{T_r}(c_r,z_r,q_r)$ is simulated as follow.
Assuming $T_1$ knows the color $c \in \mathcal{C}$ of block ${b\otimes z_r}$ (which is possible as its position in all directions modulo $2r$ is stored in its head state), we can define a reference starting position to explore the neighbouring blocks by ordering the directions of $\boule{3}{1}$.
Moreover, this ordering allows to decide a depth-first exploration of the blocks ${b\otimes z'}$ with $z' \in \boule{3}{r}$ passing through each block at most seven times.
Once the exploration done, back in block ${b\otimes z}$, the head of $T_1$ contains all necessary information to compute $\delta_r$ and all that remains to do is to write the computed letter in a cross pattern on the faces of the block. 
This is possible following a eulerian path, crossing only on the center of the faces, hence the padding of 2 defined earlier. 
The 1 padding left is for $T_1$ to align itself with its next color (which is possible since it knows its current color and has computed to in which block to go next).
\end{proof}

By combining Theorem~\ref{thm:3Dunivr1rigor} and Theorem~\ref{thm:3Dradiusliberal}, we get the existence of an intrinsically universal 3D turedo for liberal simulations as expressed in the following corollary.

\begin{corollary}\label{coro:fulluniversality3D}
  ${\univsimdim{3}{\simlib}\cap\turdim{3}\neq\emptyset}$.
\end{corollary}

\section{Discussion}
\label{sec:persp}

The problem tackled in this work depends on three parameters (radius, dimension, simulation).
Our results give a rather clear picture of the simulation hierarchies in 3D, but we left several open question in the 2D case, in particular: is there a turedo in $\tur{2}{2}$ which is universal for $\tur{2}{1}$ under rigorous simulations? what if we allow liberal simulations?
Actually even Theorem~\ref{theo:noradius1univ} raises questions: the simulation impossibility makes a crucial use of non-connected seeds, does this impossibility remain if we just ask simulation of orbits starting from connected seeds?

In this work we chose the square lattice in 2D (to simplify and as we also consider the 3D case) whereas oritatami are mostly considered on the hexagonal lattice.
We don't expect any significant difference on the simulation hierarchy result by changing from square to hexagonal lattice on a given model (either turedos or oritatami).
However, it is not clear that the delay hierarchy for oritatami behaves likes the radius hierarchy for turedos.
In particular, we don't know if an analog of Theorem~\ref{thm:heatsink} holds for oritatami.
The key difference between a large radius turedo and a large delay oritatami is that the turedo can gather information locally across obstacles, while the oritatami can only probe information around that can be reached by a path of empty positions (because it can only probe by trying to position a small strand of beads).

We end this paper by suggesting the following two directions in order to better understand the gap between turedos and oritatami systems: what if we restrict turedos to 'see' only neighboring positions that can be reached through a path of $r$ empty positions? and what if we enrich oritatami systems by a more general 'magnetic' attraction law between beads where pairs of distant beads can still contribute to the total amount of attraction that the free strand at the end of the molecule is trying to maximize (let's say by a quadratic decrease with distance up to some radius)?

\bibliography{dna}

\end{document}